\documentclass[manuscript,authorversion]{acmart}
\usepackage[]{preamble}

\acmJournal{TOPC}

\title{Engineering MultiQueues}
\subtitle{Fast Relaxed Concurrent Priority Queues}

\author{Marvin Williams}
\affiliation{%
  \institution{Karlsruhe Institute of Technology}
  \city{Karlsruhe}
  \country{Germany}
}
\email{williams@kit.edu}
\orcid{0000-0002-1990-0899}

\author{Peter Sanders}
\affiliation{%
  \institution{Karlsruhe Institute of Technology}
  \city{Karlsruhe}
  \country{Germany}
}
\email{sanders@kit.edu}
\orcid{0000-0003-3330-9349}

\renewcommand{\shortauthors}{Williams, et al.}

\begin{document}

\begin{abstract}
	Priority queues are used in a wide range of applications, including prioritized online scheduling, discrete event simulation, and greedy algorithms.
	In parallel settings, classical priority queues often become a severe bottleneck, resulting in low throughput.
	Consequently, there has been significant interest in concurrent priority queues with relaxed semantics.
	In this article, we present the \emph{MultiQueue}, a flexible approach to relaxed priority queues that uses multiple internal sequential priority queues.
	The scalability of the MultiQueue is enhanced by buffering elements, batching operations on the internal queues, and optimizing access patterns for high cache locality.
	We investigate the complementary quality criteria of \emph{rank error}, which measures how close deleted elements are to the global minimum, and \emph{delay}, which quantifies how many smaller elements were deleted before a given element.
	Extensive experimental evaluation shows that the MultiQueue outperforms competing approaches across several benchmarks.
	This includes shortest-path and branch-and-bound benchmarks that resemble real applications.
	Moreover, the MultiQueue can be configured easily to balance throughput and quality according to the application's requirements.

	We employ a seemingly paradoxical technique of ``wait-free locking'' that might be of broader interest for converting sequential data structures into relaxed concurrent data structures.
\end{abstract}



\maketitle

\begin{acks}
	This project has received funding from the \grantsponsor{erc}{European Research Council (ERC)}{} under the European Union’s Horizon 2020 research and innovation program (grant agreement No. \grantnum{erc}{882500}).\\
	\includegraphics[width=4cm]{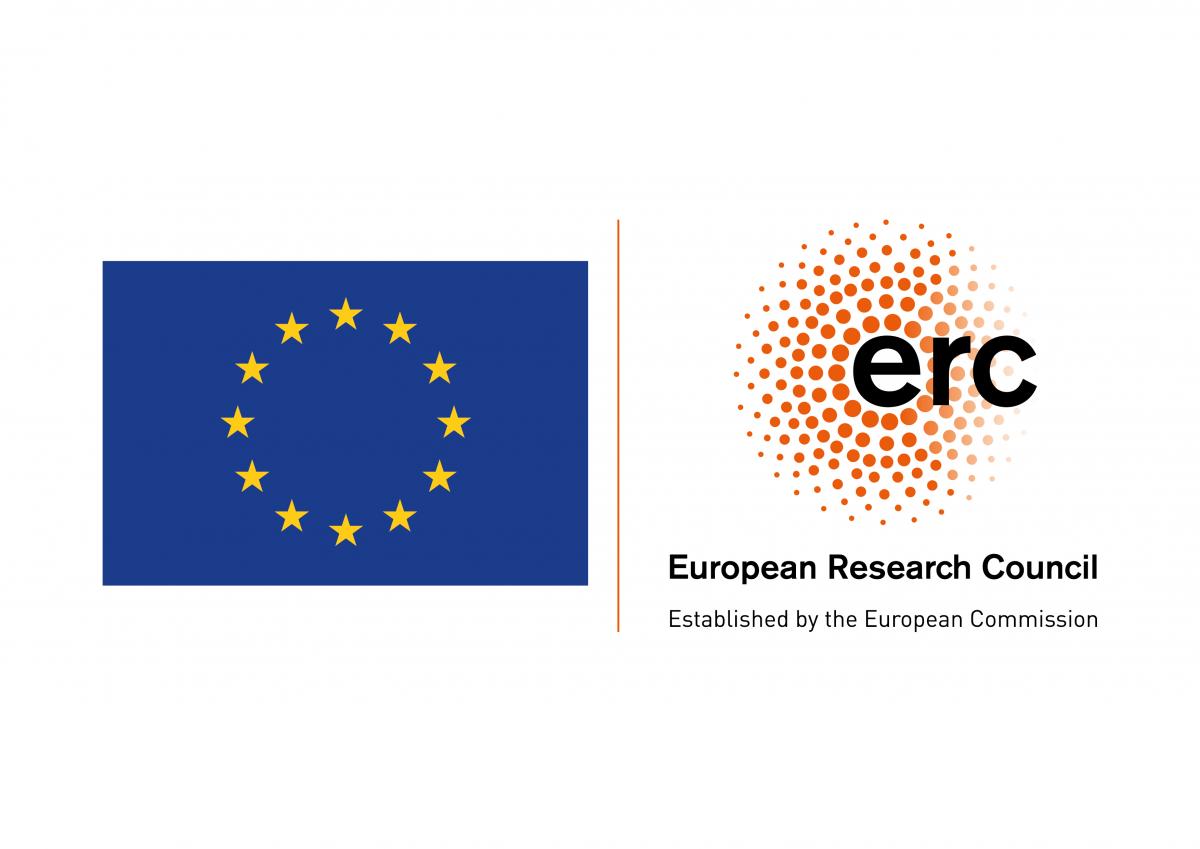}
\end{acks}

\clearpage
\section{Introduction}\label{intro}
The Priority Queue (PQ) is a fundamental\footnote{It is covered in most introductory algorithm lectures and textbooks.} abstract data type that manages a dynamic collection of elements with associated priorities.
PQs support two basic operations: inserting an element and deleting the smallest element.
They are central to many algorithms that benefit from dynamic reordering of operations, including job scheduling, graph searches, discrete event simulation, and best-first branch-and-bound searches.

Sequential PQs become a significant bottleneck in modern multi-core systems when used in parallel versions of these algorithms.
PQs that permit the concurrent insertion and deletion of elements aim to address this issue but encounter semantic ambiguities not present in sequential PQs.
These ambiguities arise when two threads attempt to delete the smallest element simultaneously, or when a deletion happens at the same time as the insertion of a new smallest element.
\emph{Linearizability}~\cite{herlihyLinearizabilityCorrectnessCondition1990} is a strict consistency criterion for concurrent data structures that can resolve these ambiguities.
Linearizability is a highly desirable property, and several linearizable PQs have been proposed (e.g., \cite{lindenSkiplistBasedConcurrentPriority2013, calciuAdaptivePriorityQueue2014,rukundoTSLQueueEfficientLockFree2021}).
Unfortunately, these PQs suffer from limited scalability due to inherent contention on the smallest element \cite{ellenInherentSequentialityConcurrent2012}.
Moreover, parallel applications must handle non-optimal deletions, even when using linearizable PQs.
To illustrate this, consider a parallel version of Dijkstra's algorithm for shortest paths, where multiple threads concurrently delete nodes from the PQ.
Unlike the sequential version, a node cannot be settled upon deletion from the PQ because another thread might simultaneously process a node yielding a shorter path to that node.
This observation motivates the \emph{relaxation} of the semantics of concurrent PQs, intentionally deleting non-optimal elements to improve throughput and scalability.
The idea is that the increased throughput offsets the overhead of performing redundant or unnecessary work.
Quality metrics quantify the degree of relaxation and can help assess the trade-off between quality and performance of a relaxed PQ.

In this article, we present the \emph{MultiQueue}, a relaxed concurrent priority queue.
We also evaluate it empirically and explore practical optimizations of its design.
The MultiQueue distributes insertions and deletions across multiple internal (non-relaxed) priority queues to reduce contention.
While the MultiQueue has a large design space, we focus on balancing scalability and quality.
Specifically, our basic MultiQueue maintains a constant factor more queues than threads, inserts elements into random queues, and deletes elements from the better of two randomly chosen queues.
To support arbitrary internal queues, these queues are protected by mutual exclusion locks.
Despite the locks, the basic MultiQueue is probabilistically wait-free.
It proves to be highly robust to different input distributions and operation sequences, but it has limited scalability due to poor data locality.
We address this issue by using \emph{buffers} for each internal queue and making threads \emph{stick} to a subset of queues for multiple consecutive operations.

The article is organized as follows:
In Section~\ref{preliminaries}, we introduce notations and definitions used throughout this article.
This includes the priority queue abstract data structure as well as the quality metrics \emph{rank error} and \emph{delay}.
After discussing previous work on relaxed concurrent PQs and related topics in Section~\ref{related-work}, we present the MultiQueue as the main contribution of this article in Section~\ref{multiqueue}.
In Section~\ref{analysis}, we analyze the theoretical runtime and quality of the MultiQueue in terms of rank error and delay.
We show that operations on the MultiQueue are nearly as fast as on sequential PQs, scaling linearly with the number of threads $p$.
Assuming a random data distribution, the rank errors and delays of the MultiQueue are expected linear in $p$ and in $\BigO(p\log p)$ with high probability (with parameter $p$).
Additionally, we analyze branch-and-bound algorithms using relaxed PQs.
In Section~\ref{improvements}, we introduce various orthogonal enhancements to the MultiQueue to increase throughput.
In Section~\ref{termination}, we discuss the problem of detecting termination of concurrent algorithms that finish when the PQ runs empty, such as Dijkstra's algorithm.
In Section~\ref{experiments}, we present an experimental evaluation of the MultiQueue, including a comprehensive comparison with alternative approaches found in the literature.
We benchmark synthetic workloads to measure both maximum throughput and quality, as well as an adaptation of Dijkstra's algorithm for shortest paths and a branch-and-bound algorithm for the knapsack problem.
Section~\ref{conclusion} concludes the article and discusses further research regarding the techniques used in the MultiQueue.

\section{Preliminaries}\label{preliminaries}
Given a universe of elements $\mathcal{U}$ and a strict weak ordering $<$ on $\mathcal{U}$, a \emph{priority queue} (PQ) is a set of elements $\mathcal{Q}\subseteq \mathcal{U}$ that supports the following operations:
\begin{description}
	\item[\Op{insert}] Insert an element from $\mathcal{U}$ into $\mathcal{Q}$.
	\item[\Op{delete}] Remove and return a minimal\footnote{\emph{Max-PQs} that delete a maximal element work analogously.} element from $\mathcal{Q}$ according to $<$.
\end{description}
The \emph{rank} of an element $e\in \mathcal{U}$ with respect to a PQ $\mathcal{Q}$ is defined as $r(e) \coloneq |\{e' \in \mathcal{Q} : e' < e\}| + 1$.
For practical purposes, we also define the special element $\Empty$ that is greater than all elements in $\mathcal{U}$ and cannot be inserted.
It is returned by the \Op{delete} operation in case the PQ is empty.
\emph{Concurrent} PQs allow multiple threads to call the \Op{insert} and \Op{delete} operations concurrently.
Concurrent PQs introduce semantic ambiguities that do not arise in sequential PQs.
For example, if multiple threads call \Op{delete} simultaneously, it is unclear which elements should be removed.
Similarly, if a \Op{delete} operation occurs concurrently with an \Op{insert} operation that inserts a new smallest element, it is unclear whether the new element should be considered for deletion.
Linearizability \cite{herlihyLinearizabilityCorrectnessCondition1990} is a widely used consistency criterion for concurrent data structures in relation to sequential semantics.
It requires that each operation appears to take effect instantaneously at some point between its invocation and completion.
However, this guarantee inevitably leads to high contention and limited scalability in linearizable PQ implementations \cite{attiyaLawsOrderExpensive2011}.
To overcome this limitation, the semantics of concurrent PQs can be \emph{relaxed}, allowing the \Op{delete} operation to delete a non-minimal element or \emph{fail} (i.e., return $\Empty$ without deleting an element) even when the PQ is not empty.
Notably, this relaxation imposes no restrictions on when deletions are allowed to fail, making it impossible to determine whether the PQ is truly empty if a deletion fails.
Most implementations of relaxed PQs therefore provide additional guarantees for the \Op{delete} operation.
For example, linearizable PQs are guaranteed to be empty if a \Op{delete} operation, invoked after the completion of the last \Op{insert} operation, fails.

Let $e$ be the element removed by a deletion.
A natural quality metric for the deletion then is the \emph{rank error}, defined as $r(e)-1$, the number of elements in the PQ smaller than $e$.
A complementary metric is the \emph{delay}, defined as the number of deletions of elements greater than $e$ while $e$ is in the PQ.
For convenience, we define a failed deletion to have a rank error equal to the size of the PQ and a delay of zero.
However, it still contributes towards the delay of all elements in the PQ.
Note that a deletion with rank error $r$ delays $r$ elements by one.
Thus, the sum of all delays is equal to the sum of all rank errors after all elements have been deleted.
Nonetheless, the distribution of the delays can be an important performance indicator.
For instance, if the minimum element is crucial for the progress of an algorithm, a PQ that consistently deletes the second-smallest element can lead to arbitrarily slow execution times.
While the rank error will not exceed one, the delay of the minimum element can become arbitrarily large.

An algorithm is said to be \emph{wait-free} if every thread is guaranteed to complete its operation in a bounded number of steps \cite{herlihyArtMultiprocessorProgramming2020}.
In the context of randomized algorithms, we say that an algorithm is \emph{probabilistically wait-free} if the \emph{expected} number of steps until a thread completes its operation is bounded.
Unless otherwise stated, we denote $p$ as the number of threads.
An event occurs \emph{with high probability} (w.h.p) if it occurs with probability at least $1-p^{-a}$ for some constant $a \geq 1$.

\section{Related Work}\label{related-work}

This journal article has a long history.
We introduced the basic idea to use the \emph{two-choice paradigm} for concurrent PQs in the context of a preprint and SPAA brief announcement in 2014 \cite{rihaniMultiQueuesSimplerFaster2014,rihaniMultiQueuesSimpleRelaxed2015}.
A ``proper'' conference paper \cite{williamsEngineeringMultiQueuesFast2021} with considerable improvements like buffering and stickiness took until 2021.
This journal article considers more stickiness variants, greatly expands the experiments (more benchmarks, machines, and competitors), discusses termination detection, and applies relaxed priority queues to parallel branch-and-bound algorithms.
During this long time period, the set of authors changed.

There is now also some follow-up work building on previous publications on the MultiQueue.
\citeauthor{alistarhPowerChoicePriority2017} \cite{alistarhPowerChoicePriority2017,alistarhDistributionallyLinearizableData2018} showed that the expected rank errors are in $\BigO(p)$ and that the expected maximum rank errors are in $\BigO(p\log p)$ if no elements smaller than a previously deleted element are inserted.
\citeauthor{walzerSimpleExactAnalysis2024}~\cite{walzerSimpleExactAnalysis2024} later generalized these results and made them more precise.
\citeauthor{postnikovaMultiqueuesCanBe2022} \cite{postnikovaMultiqueuesCanBe2022} proposed the \emph{stealing MultiQueue} (SMQ) that has thread-local queues but occasionally steals a batch of elements from a random thread.
This is conceptually similar to stickiness in the MultiQueue.
The \emph{Multi Bucket Queue}~\cite{zhangMultiBucketQueues2024} enhances the efficiency of the MultiQueue design by grouping elements with the same priority.
The Multi Bucket Queue additionally employs thread-local buffers to improve throughput.

There is a vast amount of work on sequential priority queues, focusing on different sets of operations and element types.
Binary Heaps \cite{williamsHeapsort1964} are a popular choice in practice, as they are conceptually simple and have logarithmic worst-case running times in the number of elements, with favorable constant factors (at least when fitting in the cache).
A natural generalization of binary heaps are $k$-ary heaps \cite{johnsonPriorityQueuesUpdate1975}---complete $k$-ary trees implicitly represented by an array in which elements are sorted along root-to-leaf paths.
We also adopt heaps in our implementation and note that known PQs with better cache efficiency \cite{oneilLogstructuredMergetreeLSMtree1996,sandersFastPriorityQueues2001} are problematic in our case, as they may occasionally require locking a PQ for a long time.

The oldest ancestor of the MultiQueue is the branch-and-bound algorithm by \citeauthor{karpRandomizedParallelAlgorithms1993} \cite{karpRandomizedParallelAlgorithms1993}.
In their design, each processor owns a local PQ from which it deletes. Insertions go to the PQ of a random other processor.
They show that, assuming synchronous execution, their algorithm expands at most
$2e\approx 5.44$ times as many nodes as the sequential algorithm in expectation.
However, this result does not transfer to the asynchronous setting, as shown by \citeauthor{alistarhPowerChoicePriority2017} \cite{alistarhPowerChoicePriority2017}.
Our measurements confirm empirically that this PQ exhibits rank errors growing not only with $p$ but also with the number of operations.

\citeauthor{sandersRandomizedPriorityQueues1998} \cite{sandersFastPriorityQueues1995,sandersRandomizedPriorityQueues1998} uses a similar design for a bulk-parallel priority queue.
However, the globally best elements are determined using a fast parallel selection algorithm.
This is asymptotically highly efficient but requires global synchronizations that are undesirable in many concurrent applications.
The paper outlines an approach to make bulk-parallel PQs asynchronous by using an asynchronously managed buffer data structure.
However, when this buffer runs empty, $\Theta(p)$ threads need to be activated
to refill the buffer. This seems expensive on current architectures.%
\footnote{For example, in a concurrent hash table \cite{maierConcurrentHashTables2018},  reorganization is supported by having a pool of dedicated hardware threads. But tying up such a scarce resource considerably constrains the application. On the other hand, spawning new threads would incur prohibitive operating system overhead.}

\citeauthor{hubschle-schneiderCommunicationEfficientAlgorithms2016} \cite{hubschle-schneiderCommunicationEfficientAlgorithms2016} improve locality of bulk-parallel PQs by also inserting in the processor's local PQ.
By using a search tree for the local PQs, they can efficiently find the globally best elements in bulk-delete operations.
Since the deleted elements may be distributed unevenly across the processors, explicit data redistribution is required.

\emph{Skip lists} \cite{pughSkipListsProbabilistic1990} have been intensively studied as a foundation for linearizable, lock-free concurrent priority queues (e.g., \cite{shavitSkiplistbasedConcurrentPriority2000,lindenSkiplistBasedConcurrentPriority2013,braginskyCBPQHighPerformance2016}).
They organize elements into a hierarchy of sorted linked lists, where each list is a random sample of the list at the level below, and the lowest level contains all elements.
This hierarchical structure accelerates search operations.
However, linearizable PQs based on skip lists generally suffer from high contention, which limits their scalability.
\citeauthor{rukundoTSLQueueEfficientLockFree2021} \cite{rukundoTSLQueueEfficientLockFree2021} propose the \emph{TSLQueue}, a lock-free linearizable PQ based on search trees and a sorted list.
The TSLQueue empirically outperforms current skip-list-based linearizable PQs.
However, the TSLQueue also suffers from bottlenecks and the worst-case performance is unclear, as its internal search tree can become unbalanced.

The \emph{Spraylist} \cite{alistarhSprayListScalableRelaxed2015} is are relaxed concurrent PQ based on skip lists that mitigates the contention at the list head by \emph{spraying} \Op{delete} operations to elements close to the head.
The spray is a biased random walk down the levels of the underlying skip list.
Insertions can still cause heavy contention.
The spraying yields rank errors in $\BigO(p\log^3 p)$ with high probability.
\citeauthor{wimmerLockfreeKLSMRelaxed2015} \cite{wimmerLockfreeKLSMRelaxed2015} propose the $k$-LSM, a relaxed concurrent PQ combining thread-local PQs and a global concurrent PQ, both implemented as log-structured merge-trees (LSMs) \cite{oneilLogstructuredMergetreeLSMtree1996}.
The local queues are bounded in size, and act as buffers for insertions and for elements ``stolen'' from other local queues.
The analysis of the paper arrives at a rank error bound of $p\cdot k$ for buffers of size $k$ \cite{wimmerLockfreeKLSMRelaxed2015}.
However, it does not account for the possibly long merge-operations of the shared LSM, during which parts of it are inaccessible.
The \emph{contention avoiding priority queue} (CA-PQ) \cite{sagonasContentionAvoidingConcurrent2017} by \citeauthor{sagonasContentionAvoidingConcurrent2017} uses a global concurrent skip list and thread-local insertion and deletion buffers.
The key idea is to dynamically switch between the global PQ and the local buffers depending on the contention.
Under low contention, threads access the global PQ directly, and switch to the local buffers under high contention.
This technique combines high accuracy in non-contended situations with high throughput under contention.
Each thread accesses the global PQ after $m \in \Theta(p)$ operations on its local buffer to ensure that the quality does not deteriorate.
The paper shows that a rank error in $\BigO(p^2)$ is guaranteed for every $m$-th \Op{delete} operation of each thread.\footnote{Such a guarantee could also be achieved with a simplistic version of the MultiQueue in which each thread inserts and deletes from a fixed queue except that every $m$ steps a thread scans all queues for the globally smallest element.}
Unfortunately, no guarantees have been proven for the remaining fraction of $1-1/m$ operations.

Another approach to low-overhead concurrent priority schedulers is grouping tasks into buckets such that tasks in the same bucket have the same priority.
This approach allows for highly efficient data structures for managing the buckets, but is only applicable if the number of ``active'' priorities is small.
One such application is Dijkstra's algorithm with small integer edge weights.
Julienne~\cite{dhulipalaJulienneFrameworkParallel2017} is a framework for parallel graph algorithms where vertices are grouped into priority buckets.
The buckets are processed in increasing order with a synchronization barrier between each bucket.
With monotonically increasing elements, this is equivalent to sequential execution and thus work-efficient.
The work-efficiency comes at the cost of poor parallelism and scalability if the buckets are small or some threads take much longer than others to reach the barrier.
\emph{OBIM}~\cite{nguyenLightweightInfrastructureGraph2013} and \emph{PMOD}~\cite{yesilUnderstandingPrioritybasedScheduling2019} maintain a global priority bucket map, but each thread works on its own version of the map asynchronously.
In OBIM, threads update their buckets from the global map when they encounter an empty bucket.
PMOD enhances OBIM by merging and splitting priority buckets dynamically, if they become too small or too large, respectively.
Both OBIM and PMOD are highly scalable, but the quality of the processed elements can degrade over time \cite{zhangMultiBucketQueues2024}.

\citeauthor{henzingerQuantitativeRelaxationConcurrent2013}~\cite{henzingerQuantitativeRelaxationConcurrent2013} give a general framework for quantitative relaxation in the context of concurrent data structures.
The rank error corresponds to the \emph{out-of-order} relaxation in the framework, and the delay is related to the \emph{lateness} relaxation.
However, the lateness only considers the number of \Op{delete} operations until a smallest element is deleted; delay is non-trivial to express within their framework.
In a follow-up paper, they propose a data structure similar to the MultiQueue for a relaxed FIFO queue \cite{haasDistributedQueuesShared2013}.
\citeauthor{rukundoMonotonicallyRelaxingConcurrent2019}~\cite{rukundoMonotonicallyRelaxingConcurrent2019} introduce the more general concept of \emph{semantic relaxation}.
The idea is to build scalable, relaxed versions of sequential data structures by using multiple data structure instances, and threads reuse the same instance for a number of consecutive operations.

The two-choice paradigm is known to be a very powerful concept for load balancing without priorities \cite{dietzfelbingerSimpleEfficientShared1993,azarBalancedAllocations1999,berenbrinkBalancedAllocationsHeavily2006,mitzenmacherPowerTwoRandom2001}.

\section{Basic MultiQueue}\label{multiqueue}
At the core of the MultiQueue is an array of sequential non-relaxed PQs, each protected by a mutual exclusion lock\footnote{One can also use thread-safe (linearizable) PQs to avoid locking.}.
The number of these \emph{internal} PQs is proportional to the maximum number of threads $p$ that may access the MultiQueue concurrently, resulting in $c\cdot p$ internal PQs for a \emph{queue factor} $c > 1$.
To insert an element, the MultiQueue selects an internal PQ randomly and attempts to lock it, repeating this process until locking succeeds.
Once a PQ is locked, the element is inserted into it and the lock is released.
Using the same strategy for deletions would cause rank errors to degrade not only with $p$ but also with the number of operations (see Figure \ref{fig:pop-pqs} and \cite{alistarhPowerChoicePriority2017}).
To counteract this, the \Op{delete} operation requires more effort: instead of selecting a single PQ, it selects $d\geq 2$ PQs randomly and attempts to lock the one containing the smallest element among them.
Otherwise, deletions follow the same procedure as insertions.

\begin{figure}
	\includegraphics{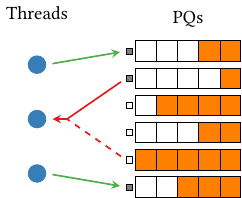}
	\caption{
		Schematic view of the MultiQueue data structure with three threads, a queue factor of $c=2$ and $d=2$ deletion candidates.
		A gray square in front of a PQ indicates that the respective PQ is locked.
		Green and red arrows represent insertions and deletions, respectively.
		The dashed red line indicates the deletion candidate that is not deleted from.
	}
	\label{fig:multiqueue}
\end{figure}
\begin{figure}
	\begin{subfigure}[t]{0.37\textwidth}
		\begin{procedure}[H]
			\DontPrintSemicolon
			\TitleOfAlgo{insert}
			\SetKwData{Array}{$A$}
			\SetKwData{Index}{$I$}
			\SetKwData{Elem}{$e$}
			\SetKwFunction{Rand}{random}
			\SetKwFunction{TryLock}{tryLock}
			\SetKwFunction{Unlock}{unlock}
			\SetKwFunction{Insert}{insert}
			\KwData{Element to insert \Elem}
			\BlankLine
			\Repeat{
				\TryLock{$\Array[\Index]$}
			}{
				$\Index \gets \Rand{$0, c\cdot p - 1$}$\;
			}
			$\Array[\Index].\Insert{\Elem}$\;
			\Unlock{$\Array[\Index]$}
		\end{procedure}
	\end{subfigure}%
	\begin{subfigure}[t]{0.63\textwidth}
		\begin{procedure}[H]
			\DontPrintSemicolon
			\TitleOfAlgo{delete}
			\SetKwData{Array}{$A$}
			\SetKwData{IndexA}{$I_1$}
			\SetKwData{IndexB}{$I_2$}
			\SetKwData{Elem}{$e$}
			\SetKwFunction{Rand}{random}
			\SetKwFunction{TryLock}{tryLock}
			\SetKwFunction{Delete}{delete}
			\SetKwFunction{Min}{min}
			\SetKwFunction{IsEmpty}{empty}
			\KwResult{Removed element or $\Empty$}
			\BlankLine
			\Repeat{
				\TryLock{$\Array[\IndexA]$}
			}{
				$\IndexA \gets \Rand{$0, c\cdot p - 1$}$\;
				$\IndexB \gets \Rand{$0, c\cdot p - 1$}$\;
				\uIf{$\Array[\IndexB].\Min{} < \Array[\IndexA].\Min{}$}{%
					$\IndexA \gets \IndexB$\;
				}
			}
			$\Elem \gets \Array[\IndexA].\Delete{}$\;
			\Unlock{$\Array[\IndexA]$}\;
			\Return \Elem
		\end{procedure}
	\end{subfigure}%
	\caption{
		Pseudocode for the \Op{insert} and \Op{delete} operations with $d=2$.
		The internal PQs are stored in array $A$.
		The \Op{min} function returns a smallest element or $\Empty$ if the PQ is empty.\label{multiqueue-op-code}}
\end{figure}
Figure~\ref{fig:multiqueue} depicts the MultiQueue structure, and Figure~\ref{multiqueue-op-code} gives pseudocode for the \Op{insert} and \Op{delete} operations.
Note that the operations never wait for a lock to be released.
There are always unlocked PQs available, since at most $p$ PQs are locked at any given time.
We emphasize that for deletions, only the ``better'' PQ is locked, which requires that reading the smallest elements of a PQ is thread-safe.
For implementation details, refer to Section~\ref{experiments:implementation}.
The \Op{delete} operation returns $\Empty$ if both selected PQs are empty, optimistically reporting that the MultiQueue is likely empty without considering additional PQs.
A practical implementation might make additional efforts to search for elements when encountering empty PQs to facilitate termination detection.

The motivation for considering multiple PQs for deletion stems from previous work on randomized resource allocation, particularly the balls-into-bins process.
In this process, $n$ balls are successively placed into one of $p$ bins.
The number of balls in a bin is the \emph{load} of this bin, and the normalized load is defined as the load minus the average load across all bins.
Assuming $n\gg p$, the normalized maximum load at the end of this process is in $\Theta(\sqrt{n\log p/p})$ with high probability, provided that each ball is assigned to a bin uniformly at random \cite{raabBallsBinsSimple1998}.
If, instead, the least-loaded of $d\geq 2$ randomly chosen bins is selected for each ball, the normalized maximum load is in $\BigO(\log \log p)$ with high probability \cite{berenbrinkBalancedAllocationsHeavily2006}, yielding an exponential improvement known as the \emph{two-choice paradigm} \cite{mitzenmacherPowerTwoRandom2001}.
\citeauthor{alistarhPowerChoicePriority2017} \cite{alistarhPowerChoicePriority2017} propose a protocol in which, with probability $0 < \beta \leq 1$, two PQs are considered for deletion, and only one otherwise, in order to achieve a smaller effective value of $d=1+\beta$.

In the time window during deletions, between comparing the smallest elements of the selected PQs to determine which PQ to lock and successfully locking it, the comparison can become \emph{stale}, meaning that the smallest element in the selected PQs has changed.
A comparison becomes stale if another thread deletes the smallest element from one of the candidate PQs or inserts a new smallest element into one of them during this time window.
Due to stale comparisons, the wrong PQ may be locked, resulting in a ``bad'' deletion or even a failed deletion, despite the other selected PQs being nonempty.\footnote{This issue does not occur when comparing and locking the PQs atomically, e.g., using transactional memory.}
To avoid stale comparisons, one can lock both PQs before comparing them.
This approach introduces significant overhead and requires a queue factor of $c>2$ to maintain wait-freedom.
Alternatively, one can recompare the PQs after acquiring the lock and retry the operation if the comparison yields a different outcome than before.
Preliminary experiments showed no significant quality improvements from preventing stale comparisons, suggesting that they are rare or have a negligible impact in practice.
\citeauthor{alistarhDistributionallyLinearizableData2018} confirm, for a simplified MultiQueue process, that the quality of deleted elements does not deteriorate due to stale comparisons \cite{alistarhDistributionallyLinearizableData2018}.

\section{Theoretical Analysis}\label{analysis}
In this section, we analyze the theoretical performance and quality of the MultiQueue.
We also analyze the running time of branch-and-bound algorithms when using relaxed concurrent PQs instead of sequential PQs.

\subsection{Running Time}\label{analysis:runtime}
We analyze the asymptotic running time of the \Op{insert} and \Op{delete} operations within a realistic asynchronous shared memory model (specifically, the aCRQW model \cite[Section~2.4.1]{sandersSequentialParallelAlgorithms2019}).
In this model, the time required for performing $p$ contended writes to the same machine word is in $\BigO(p)$.

\begin{theorem}\label{thm:runtime}
	The expected time for a thread to acquire a lock during the \Op{insert} and \Op{delete} operations is in $\BigO(1)$.
\end{theorem}
\begin{proof}
	At most $p-1$ internal PQs are locked by other threads at any given time.
	Hence, the worst-case probability $r$ that a randomly selected PQ is locked is
	\[
		r\coloneq\frac{p-1}{cp} \leq \frac{1}{c} < 1.
	\]
	Since the probabilities of subsequent locking attempts to succeed are independent, the probability that a thread needs $i$ attempts to find an unlocked queue is $r^{i-1}(1-r)$.
	The expected number of attempts to find an unlocked queue is therefore
	\[
		\sum_{i=1}^\infty{ir^{i-1}(1-r)}=\frac{1}{1-r}\leq \frac{1}{1-\frac{1}{c}}=\frac{c}{c-1}\in \BigO(1).
	\]
	Since the expected number of threads contending to lock a PQ is constant, the same bound holds for the expected execution time.
\end{proof}
We obtain the following bounds for the comparison-based model and for integer keys, respectively:
\begin{corollary}
	A MultiQueue with binary heaps needs constant average insertion time and expected time $\BigO(\log n)$ per operation for worst-case operation sequences.
	With van Emde Boas trees \cite{vanemdeboasDesignImplementationEfficient1976,mehlhornBoundedOrderedDictionaries1990} and integer keys in $\{0,\ldots,U-1\}$, an expected time of $\BigO(\log\log U)$ per operation is achieved.
\end{corollary}

Theorem~\ref{thm:runtime} holds even when other threads are suspended, since in the worst case, each suspended thread holds one lock.
The overhead of attempting to acquire these locks is already accounted for in the analysis.
As a result, each thread is expected to make progress in a bounded number of steps, making the MultiQueue probabilistically wait-free.
Note that ``progress'' in this context refers to completing an \Op{insert} or a (potentially failing) \Op{delete} operation, rather than the completion of an application using the MultiQueue.
Specifically, if a thread is suspended while holding a lock on an internal PQ, the elements in that PQ become inaccessible to other threads.
This can potentially delay an application's progress indefinitely, rendering it not wait-free.

\subsection{Quality}\label{analysis:quality}
A general quality analysis of the MultiQueue is still an open problem.
\citeauthor{alistarhPowerChoicePriority2017} were the first to provide an asymptotic analysis of the rank error in a simplified sequential MultiQueue process, where no elements smaller than any previously deleted element are inserted.
They show that the expected rank error is in $\BigO(p)$ and the expected maximum rank error is in $\BigO(p\log p)$ for any number of deletions.
They later showed that these bounds hold even in a concurrent setting with stale comparisons \cite{alistarhPowerChoicePriority2017,alistarhDistributionallyLinearizableData2018}.
\citeauthor{walzerSimpleExactAnalysis2024} \cite{walzerSimpleExactAnalysis2024} present an exact analysis of the rank error distribution for the simplified (sequential) MultiQueue process with more general deletion strategies.
The expected long-term rank error for $d=2$ turns out to be $\tfrac{5}{6}cp-1+\tfrac{1}{6cp}$.
Both \citeauthor{alistarhPowerChoicePriority2017} \cite{alistarhPowerChoicePriority2017} and \citeauthor{walzerSimpleExactAnalysis2024} \cite{walzerSimpleExactAnalysis2024} confirm that the distribution of elements to the PQs does not stabilize for $d=1$, leading to diverging rank errors (see Figure~\ref{fig:pop-pqs}).

In this section, we will explain how \emph{rank errors} and \emph{delays} can be estimated under simplified but intuitive assumptions.
In the following, we only consider the sequential execution of operations.
We assume that, at any given time, each element is equally likely to be in any queue.
This assumption holds at least until the first \Op{delete} operation is executed.
The reasoning behind this assumption is as follows:
Insertions place elements into random queues, which moves the system toward a uniform distribution of elements.
Deletions, on the other hand, select PQs with small elements more frequently than those with large elements, thus controlling the deviation among the queues.
\citeauthor{walzerSimpleExactAnalysis2024} \cite{walzerSimpleExactAnalysis2024} show that this assumption does not generally hold but conjecture that the distribution of elements will never be ``too far off'' the uniform distribution, even for arbitrary operation sequences.

\paragraph{Rank error}
Let $R$ be the random variable indicating the rank error of a \Op{delete} operation.
This means that there are $R$ elements smaller than the deleted element, and all of those elements are in PQs other than the ones considered for the \Op{delete} operation.
Under our assumptions, the probability that a specific element is in one of the selected PQs is $s\coloneq\tfrac{d}{cp}$, independent of the other elements.
Thus, $R$ follows a geometric distribution with parameter $s$:
\[
	\textrm{P} (R = i) = \left(1-s\right)^is.
\]
Since $\textrm{E}[R] = \tfrac{1}{s}-1$, the expected rank error is in $\BigO(p)$.

The probability to delete an element with rank error $i$ or greater is given by
\[
	\textrm{P} (R \geq i) = \left(1-s\right)^i.
\]
For a rank error $l=\tfrac{1}{s}\log p\in\BigO(p\log p)$, we have $\textrm{P}(R \geq l) = \left(1-s\right)^l \approx e^{-s\cdot l} = p^{-1}$.
Therefore, rank errors are in $\BigO(p\log p)$ with high probability.

\paragraph{Delay}
Since deleting an element with rank error $i$ delays $i$ elements, the sum of all delays is equal to the sum of all rank errors after all elements are deleted.
Therefore, the expected delay is equal to the expected rank error.
Consider an element $e$ with rank $i+1$.
The probability for a deletion to remove an element larger than $e$, thus increasing the delay of $e$ by one, is $\textrm{P} (R > i)$.
With probability $\textrm{P} (R = i)$, the \Op{delete} operation removes $e$.
Hence, the probability that a \Op{delete} operation either increases the delay of $e$ or removes $e$ is
\[
	\textrm{P} (R > i) + \textrm{P} (R = i) = \textrm{P} (R>=i).
\]
The conditional probability that a \Op{delete} operation, which would otherwise delay $e$, instead removes $e$ is
\[
	\textrm{P} (R = i \mid R \geq i) = \frac{\textrm{P} (R = i)}{\textrm{P} (R \geq i)} = \frac{(1-s)^is}{(1-s)^i}=s.
\]
Since this probability is independent of $i$, the delay also follows a geometric distribution.
In fact, the delay distribution is identical to the rank error distribution.
Consequently, the expected delay is in $\BigO(p)$ and, with high probability, the delay is in $\BigO(p\log p)$.

\subsection{Branch-and-bound}
In this section, we analyze the run-time of the best-first branch-and-bound scheme using a relaxed concurrent PQ.
Our approach is similar to the analysis for bulk-parallel priority queues \cite{sandersFastPriorityQueues1995}, but leverages the expected delay to estimate the overhead introduced by the relaxation.

Branch-and-bound algorithms solve optimization problems by exploring and pruning the search tree of possible (partial) solutions.
Given the currently best solution as an upper bound for the optimal (minimal) solution, partial solutions with lower bounds larger than the currently best solution are pruned.
The best-first heuristic selects the partial solution with the smallest lower bound for further exploration.
This heuristic is typically implemented by a PQ managing partial solutions ordered by their lower bound.

Let $H$ be the search tree of partial solutions, and let $H_{\leq}\subseteq H$ be the subset of partial solutions whose lower bounds are not larger than the optimal solution $s$.
A sequential best-first branch-and-bound algorithm might explore all of $H_{\leq}$ while pruning all partial solutions in $H_{>}\coloneq H\setminus H_{\leq}$.
Partial solutions in $H_{>}$ cannot be explored before $s$ is found, since their lower bounds are larger than $s$.
After $s$ is found, all partial solutions in $H_{>}$ are pruned.
However, with relaxed PQs, some nodes from $H_{>}$ might be explored before they can be pruned.
\begin{figure}
	\includegraphics[width=0.4\textwidth]{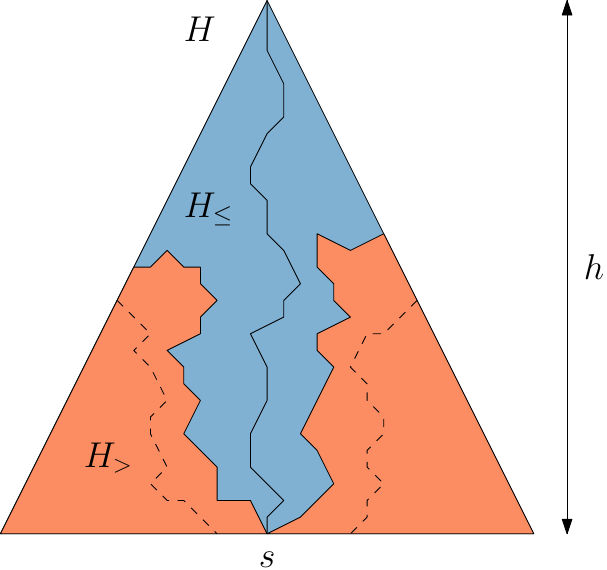}
	\caption{
		Schematic view of the search tree $H$ with height $h$ with nodes in $H_{\leq}$ and $H_{>}$.
		The path from the root to the optimal solution $s$ is highlighted.
		The dashed lines bound the area of $H_{>}$ that are explored due to the delay.
	}
	\label{fig:bnb}
\end{figure}
Figure~\ref{fig:bnb} illustrates the search tree and the nodes that are unnecessarily explored due to the delay.
Let $h$ be the length of the path $S$ from the root of $H$ to $s$ and $D$ the expected delay of the relaxed PQ.
Until $s$ is found, the PQ always contains at least one partial solution from $S$.
Each deletion from $H_{>}$ delays all partial solution on $S$ that are in the PQ at the time of the deletion.
Since each partial solution in $S$ is delayed by $D$ in expectation, at most $Dh$ nodes from $H_{>}$ are expected to be explored before $s$ is found.
In total, $n' \leq H_{\leq} + Dh$ nodes are explored in expectation.
Assuming that PQ operations and processing a node take $\BigO(\log n)$ time, the expected parallel run-time $T_\text{par}$ with $p$ threads is
\[
	T_\text{par} = \frac{T_{\text{seq}}}{p} + \BigO\left(\left(\frac{H_{\leq} + Dh}{p}\right)\log n\right).
\]
For the MultiQueue, $D$ is in $\BigO(p)$, leading to a total expected run-time of
\[
	T_\text{par} = \tfrac{T_{\text{seq}}}{p} + \BigO\left(\left(\tfrac{H_{\leq}}{p}+h\right)\log n\right).
\]
Note that this bound is essentially the same as proven for bulk-parallel PQs \cite{sandersFastPriorityQueues1995}.
However, a crucial difference is that previously, the bound applied only when execution times for node processing were always the same.

\section{Practical Improvements}\label{improvements}
In this chapter, we introduce practical improvements to the MultiQueue design as presented in Section~\ref{multiqueue}.
Even for small numbers of elements, cache efficiency is an issue for the MultiQueue.
By operating on a specific internal PQ, a thread will move the corresponding cache lines into its cache.
Most likely, this PQ will be accessed by other threads next, leading to cache misses for these threads as well as invalidation traffic due to cache coherency.
We address this issue and present techniques to significantly improve the cache efficiency of the MultiQueue in the following sections.

\subsection{Buffering}\label{multiqueue-buffering}
To reduce the average number of cache lines accessed by the \Op{insert} and \Op{delete} operations, we enhance each internal PQ with an \emph{insertion buffer} and a \emph{deletion buffer} of fixed capacity.
Generally, insertion go into the insertion buffer, and deletions are performed on the deletion buffer.
We maintain the invariants that the smallest elements are in the deletion buffer in sorted order, and if the deletion buffer is empty when trying to delete an element, the insertion buffer and the PQ are also empty.
Therefore, small elements have to be inserted into the deletion buffer directly.
If the insertion buffer is full, its elements are flushed into the PQ.
If the deletion buffer becomes empty, it is refilled with the smallest elements from both the insertion buffer and the PQ.
This is done by first flushing the insertion buffer into the PQ and then refilling the deletion buffer solely from the PQ.
\begin{figure}
	\begin{subfigure}[t]{0.5\textwidth}
		\begin{procedure}[H]
			\DontPrintSemicolon
			\TitleOfAlgo{insertWithBuffers}
			\SetKwData{PQ}{$Q$}
			\SetKwData{Elem}{$e$}
			\SetKwData{Tmp}{$e'$}
			\SetKwData{D}{$D$}
			\SetKwData{I}{$I$}
			\SetKwData{CapI}{$C_I$}
			\SetKwData{CapD}{$C_D$}
			\SetKwFunction{Max}{max}
			\KwData{Element to insert \Elem}
			\BlankLine
			\uIf{$\I = \emptyset \land \PQ = \emptyset \land |\D| < \CapD$}{
				$\D \gets \D \cup \{\Elem$\}\;
				\Return
			}
			\uIf{$\Elem < \Max{\D}$}{
				\uIf{$|\D| < \CapD$}{%
					$\D \gets \D \cup \{\Elem$\}\;
					\Return
				}
				$\Tmp \gets \Max{\D}$\;
				$\D \gets (\D \setminus \{\Tmp\}) \cup \{\Elem\}$\;
				$\Elem \gets \Tmp$\;
			}
			\lIf{$|\I| = \CapI$}{flush \I into \PQ}
			$\I \gets \I \cup \{\Elem\}$
		\end{procedure}
	\end{subfigure}%
	\begin{subfigure}[t]{0.5\textwidth}
		\begin{procedure}[H]
			\DontPrintSemicolon
			\TitleOfAlgo{deleteWithBuffers}
			\SetKwData{PQ}{$Q$}
			\SetKwData{Elem}{$e$}
			\SetKwData{D}{$D$}
			\SetKwData{I}{$I$}
			\SetKwFunction{Min}{min}
			\KwResult{Removed element or $\Empty$}
			\BlankLine
			\lIf{$\D = \emptyset$}{\Return $\Empty$}
			$\Elem \gets \Min{\D}$\;
			$\D \gets \D \setminus \{\Elem\}$\;
			\uIf{$\D = \emptyset$}{
				flush \I into \PQ\;
				refill \D from \PQ\;
			}
			\Return \Elem
		\end{procedure}
	\end{subfigure}%
	\caption{Pseudocode for inserting into and deleting from a locked PQ $Q$ with insertion buffer $I$ and deletion buffer $D$.
		The buffers have capacities $C_I$ and $C_D$, respectively.
		The \Op{max} and \Op{min} operations return a largest and smallest element or $\Empty$ if the set is empty, respectively.
		Refilling $D$ from $Q$ is done by iteratively deleting the smallest element from $Q$ and inserting it into $D$ until $D$ is full or $Q$ is empty.
	}\label{multiqueue-buffer-code}
\end{figure}
Figure~\ref{multiqueue-buffer-code} shows the pseudocode for the \Op{insertWithBuffers} and \Op{deleteWithBuffers} operations.
For implementation details of the buffers, refer to Section~\ref{experiments:implementation}.
An alternative approach to refill the deletion buffer is to refill from the PQ first, and then swap elements with the insertion buffer to re-establish the invariant that the smallest elements are in the deletion buffer.
This approach has the advantage that all interactions between the buffers and the PQ are in batches of fixed size, which can be exploited by specialized PQs (see Section~\ref{multiqueue-pqs}).
However, it incurs additional moves of elements between the buffers and the PQ.

With buffers, insertions and deletions interact with the internal PQ itself only when the buffers are full or empty.
An insertion typically reads the largest element from the deletion buffer to determine whether the element should be inserted into the deletion buffer directly, and otherwise operates only on the insertion buffer.
Deletions typically only access the deletion buffer.
On the one hand, buffering can reduce cache misses, since only one thread accesses the PQ in a bulk-fashion, exhibiting high temporal locality without interference from other threads.
Implicit heaps, for example, exhibit high locality for insertions, and at least exhibit some locality near the root and at the rightmost end of the bottom layer of the tree for deletions.
On the other hand, buffering increases the worst-case time for each operation.
If an operation locks a PQ that contains small elements for a long time, the delay of those elements, as well as the rank error of other deletions, increases.

\subsection{Cache-efficient PQs}\label{multiqueue-pqs}
Implicit tree-like data structures like binary heaps can be used to avoid cache misses from indirection within the internal PQs.
$k$-ary heaps with $k>2$ improve on binary heaps by reducing the number of cache misses to $\BigO(\log_kn)$ if $k$ elements fit into one cache line.
To fully exploit the fact that the MultiQueue accesses the internal PQ in a bulk-fashion due to buffering, data structures that directly support batch operations can also be utilized.
A promising data structure is the \emph{merging binary heap}, an adaptation of the \emph{parallel heap} by \citeauthor{deoParallelHeapOptimal1992}~\cite{deoParallelHeapOptimal1992}.
Merging binary heaps are structured like binary heaps, but each node contains a fixed number of sorted elements.
The heap invariant is that the first element of each node is not smaller than the last element of its parent node.
Insertion and deletion work similarly to ordinary binary heaps.

New nodes are inserted as a leaf.
Then, the node is merged with its parent, moving the smaller half into the parent node and the larger half into the child node.
This process is repeated with the parent of the current node until the first element of the node is not smaller than the last element of its parent or the root is reached.
When the root node is deleted, the two child nodes are merged and the smaller half is put into the parent node, the larger half into the child node which had the larger last element.
Then, the other child node is empty, and the process is repeated with this child until a leaf is reached.
Finally, the last node is moved to this leaf, and the heap invariant is restored analogously to the insertion process.
While merging binary heaps require fewer tree operations than $k$-ary heaps, they come with additional algorithmic complexity and higher worst-case access times.
In our preliminary experiments, merging binary heaps and $k$-ary heaps performed very similarly, so we decided to use the conceptually simpler $k$-ary heaps.

\subsection{Stickiness}\label{multiqueue-sticky}
Stickiness is a simple yet effective technique to increase temporal cache locality at the cost of potentially higher rank errors and delays.
The basic idea is as follows:
Threads stick to the same set of $d\geq 2$ internal PQs for $s\geq 1$ consecutive operations (the \emph{stickiness period}).
Insertions pick one PQ from the set at random.
Deletions use the PQ with the smallest element among the set.
PQs are still only locked for single operations.
Each thread replaces its set with a new set of $d$ PQs after $s$ consecutive operations or if locking fails for one operation (to avoid blocking).
We assume $d \leq c$ to ensure that there are always enough PQs for each thread to stick to.

In the \emph{simple} variant, the PQs are chosen uniformly at random, independently for each thread.
Thus, multiple threads may stick to the same PQs simultaneously, which is undesirable because of contention and cache invalidations.
To mitigate this issue, threads can leave a thread-specific mark in the lock of their most recently accessed PQ.
Then, they can detect if they try to lock a PQ that another thread recently accessed, and choose a different one to stick to.

The \emph{swap} variant helps guiding threads to currently unused PQs by maintaining a permutation of PQ indices in a global array.
Each thread is assigned $d$ fixed positions in this array, indicating the PQs to stick to.
Whenever a thread wishes to replace its set of PQs, it atomically swaps the indices in each of its array positions with random other positions.
Thus, no two threads stick to the same PQ simultaneously.
Still, threads need to acquire locks for each operation because there might be ongoing operations from other threads on the assigned PQs.
\begin{figure}
	\begin{procedure}[H]
		\DontPrintSemicolon
		\TitleOfAlgo{atomicSwap}
		\SetKwData{I}{$i$}
		\SetKwData{J}{$j$}
		\SetKwData{P}{$a_i$}
		\SetKwData{Q}{$a_j$}
		\SetKwData{A}{$A$}
		\SetKwData{Unavailable}{$-1$}
		\SetKwFunction{CAS}{compareAndSwap}
		\SetKwFunction{Exchange}{atomicExchange}
		\SetKwFunction{Rand}{random}
		\SetKw{Continue}{continue}
		\KwData{Index to swap from \I}
		\BlankLine

		$\P \gets \Exchange{$\A[\I],\Unavailable$}$\;
		\Repeat{$\Q \neq \Unavailable \land \CAS{$\A[\J],\Q,\P$}$}{
			$\J \gets \Rand{$0, c\cdot p - 1$}$\;
			$\Q \gets \A[\J]$\;
		}
		$\A[\I] \gets \Q$\;
	\end{procedure}
	\caption{Pseudocode for the atomically swapping an entry $i$ with another randomly chosen entry.
		The operation $\Op{atomicExchange}(A,v)$ atomically reads the value at $A$ and sets it to $v$, the operation $\Op{compareAndSwap}(A,e,x)$ atomically compares the value at $A$ with $e$ and sets it to $x$ if they are equal.
	}
	\label{atomic-swap}
\end{figure}
Since atomic swaps are not natively supported on most hardware, we use the algorithm given in Figure~\ref{atomic-swap} to atomically swap entries in the permutation array.
It is crucial for the correctness of the algorithm that each entry can only be invalidated (set to $-1$) by one specific thread.
The algorithm is probabilistically wait-free because valid entries are found in expected constant time.
Our preliminary experiments showed that the marking approach is inferior to the \emph{swap} variant.

Another approach we did not investigate further is to induce the assignment using a global permutation of the form $\pi(i)=i\cdot a+b\bmod m$ with parameters $a \in \mathbb{N}$ and $b \in \mathbb{N}_0$ such that $a$ and $m=c\cdot p$ are co-prime.\footnote{Since $a$ and $m$ are co-prime, there exists a unique $a^{-1}$ such that $aa^{-1}\equiv 1 \mod m$.
For $\pi(i)=\pi(i')$ it must be $ia+b\equiv i'a+b \mod m\Rightarrow ia\equiv i'a \mod m\Rightarrow iaa^{-1}\equiv i'aa^{-1} \mod m\Rightarrow i\equiv i' \mod m$.}
The parameters $a$ and $b$ are shared between all threads.
Each thread is assigned fixed unique indices $i,\ldots,i+d-1\in [0,m-1]$, and sticks to the PQs $\pi(i),\ldots,\pi(i+d-1)$.
This approach avoids the rather expensive swapping in the permutation array, but it changes the PQ assignments for all threads simultaneously.
The threads have to agree when to change the permutation, which can be challenging in practice.

\section{Termination Detection}\label{termination}
\emph{Termination detection} is the process of determining when an algorithm has completed its task.
Many algorithms involving priority queues repeatedly delete an element from the PQ, process it (thereby possibly inserting new elements), and terminate after the PQ becomes empty.
While it is trivial to detect this termination condition in sequential settings, it can be challenging in parallel settings.\footnote{
	Termination (or \emph{quiescence}) detection  is a well-studied problem in distributed systems.
	For an overview, see \cite{matochaTaxonomyDistributedTermination1998} and \cite{matternAlgorithmsDistributedTermination1987}.}
The goal is to ensure that all threads keep working until the PQ is permanently empty and then terminate.
More precisely, the termination detection algorithm should satisfy the following properties:
\begin{enumerate}
	\item Each thread must repeatedly attempt to delete an element and process it until the PQ is empty and no new elements will be inserted.
	\item If the PQ is empty and no new elements will be inserted, all threads eventually terminate.
\end{enumerate}
Checking a potential concurrent \Op{isEmpty} operation generally does not suffice (even if linearizable) because a thread may observe an empty PQ while other threads are still processing elements.
Thus, termination detection requires some kind of coordination between the threads.
If the total number of elements to process is known in advance, each thread can maintain a local counter of processed elements.
When a thread fails to delete an element, it updates a global atomic counter of processed elements.
Since this global counter never overestimates the number of processed elements, threads can terminate as soon as the counter reaches the total number of elements to process.
However, if the total number of elements to process is unknown, more sophisticated coordination is required.
Tracking the number of remaining elements to process with a global atomic counter is infeasible in high-throughput scenarios due to the high contention on that counter.
\citeauthor{maierConcurrentHashTables2018}~\cite{maierConcurrentHashTables2018} propose a more scalable version of this approach in the context of approximating the size of concurrent hash tables, where each thread counts the number of its insertions and deletions, and update a global counter sporadically.
However, this counter can become zero despite some elements remain to be processed, making it unsuitable for termination detection.
\citeauthor{ellenSNZIScalableNonZero2007}~\cite{ellenSNZIScalableNonZero2007} propose a scalable concurrent non-zero indicator for the number of remaining elements to process, incurring a small overhead per operation.
This indicator is especially useful in scenarios where the PQ becomes empty frequently.
However, relaxed concurrent PQs are most useful in scenarios where many elements are in the PQ for most of the time and processing elements is very fast, making this approach less suitable.
If there are few elements to process most of the time, faster unordered data structures might be more suitable than relaxed PQ.

Another approach is to periodically synchronize all threads and determine whether all elements have been processed.
Since global synchronization is expensive, this approach must be carefully designed to avoid excessive synchronization while ensuring that all threads terminate in a timely manner.

\begin{figure}
	\begin{procedure}[H]
		\DontPrintSemicolon
		\TitleOfAlgo{ProcessUntilEmpty}
		\SetKwBlock{Repeat}{repeat}{end}
		\SetKwData{NumThreads}{$p$}
		\SetKwData{PollingCount}{polling}
		\SetKwData{IdleCount}{idle}
		\SetKwData{PQ}{$Q$}
		\SetKwData{e}{$e$}
		\SetKwFunction{Process}{process}
		\SetKw{Goto}{goto}
		\SetKw{Break}{break}
		\SetKw{Continue}{continue}
		\SetKwFunction{Del}{delete}
		\KwData{Priority queue \PQ}
		\KwData{Number of threads \NumThreads}
		\KwData{Global atomic counters \PollingCount and \IdleCount}
		\BlankLine
		\nl \Repeat{
			\nl $\e \gets \PQ.\Del{}$\label{working}\;
			\nl \If{$\e = \Empty$}{
				\nl $\PollingCount \gets \PollingCount+1$\label{enterpolling}\;
				\nl \Repeat{
					\nl 	$\e \gets \PQ.\Del{}$\label{trypop}\;
					\nl 	\lIf{$\e \neq \Empty$}{
						\Break
					}
					\nl 	\lIf{$\PollingCount < \NumThreads$}{\Continue\label{skipidle}}
					\nl 	$\IdleCount \gets \IdleCount+1$\label{enteridle}\;
					\nl 	\While{$\PollingCount = \NumThreads$}{
						\nl 		\lIf{$\IdleCount = \NumThreads$}{
							\Return\label{terminate}
						}
					}
					\nl 	$\IdleCount \gets \IdleCount-1$\;
				}
				\nl $\PollingCount \gets \PollingCount-1$\;
			}
			\nl\tcp{process \e}\label{process}
		}
	\end{procedure}
	\caption{Pseudocode for termination detection.
		\texttt{polling} and \texttt{idle} are global atomic counters, initialized to 0.
	}
	\label{termination-code}
\end{figure}
The algorithm given in Figure~\ref{termination-code} avoids global synchronization and incurs no overhead for successful deletions, relying on the following assumption: The PQ is guaranteed to be empty if the last operation of every thread was a failed deletion and currently no insertions are in progress.
To the best of our knowledge, most relaxed PQ designs can be implemented such that this assumption holds.
Note that linearizable PQs provide an even stronger guarantee by definition: The PQ is empty if for any thread a \Op{delete} operation failed that started after the last \Op{insert} operation completed.
However, the \Op{delete} operation of the MultiQueue as given in Figure~\ref{multiqueue-op-code} does not satisfy this assumption.
We therefore extend it to retry the operation a fixed number of times upon failure, after which all internal PQs are scanned linearly for elements (locked PQs are skipped).

The algorithm detects termination by tracking the number of threads whose last operation was a failed deletion.
To prevent premature termination, the algorithm counts the threads twice, similar to the well-known \emph{double counting} mechanism for distributed systems \cite{matternAlgorithmsDistributedTermination1987}.
Threads can be in one of three states: \emph{working}, \emph{polling}, or \emph{idling}.
The number of threads in the polling and idling states are tracked using dedicated counters.
All threads begin in the working state, where they repeatedly delete elements from the PQ (Line~\ref{working}).
If a thread fails to delete an element, it transitions to the polling state (Line~\ref{enterpolling}).
While polling, the thread keeps attempting deletions until either it succeeds, or all other threads are also polling (Lines~\ref{trypop}--\ref{skipidle}).
If a deletion was successful, the thread returns to the working state.
If a polling thread detects that all threads are polling, it transitions to the idling state (Line~\ref{enteridle}).
Note that terminating at this point would be premature, since another polling thread may have just successfully deleted an element but has not yet entered to the working state.
In the idling state, a thread waits until all other threads also become idle, at which point it terminates (Line~\ref{terminate}).
However, if another thread transitions back to the working state, the idling thread re-enters the polling state.

In the following, we argue that the termination detection algorithm satisfies the properties (1) and (2).
A thread terminates if and only if it observes the \Op{idle} counter reaching $p$ in Line~\ref{terminate}.
Since a thread must decrement a counter before incrementing it again, for the \Op{idle} counter to reach $p$, all threads must be in the \Op{idle} state at the same time.
At that time, every thread must have performed a failed deletion without subsequently starting an insertion.
Given the assumption that the PQ is empty under these conditions, and no thread has an element to process, all threads can safely terminate.
Moreover, each thread attempts to delete the next element to process within a bounded number of steps (assuming uniform progress and instantaneous memory propagation) until it terminates.
If no elements remain to be processed, all threads will eventually leave the working state and never re-enter it.
Thus, the \Op{polling} counter will eventually reach $p$, and all threads will eventually transition to the \Op{idle} state and never leave it again.
Once all threads reach the \Op{idle} state, the \Op{idle} counter reaches $p$ and all threads terminate.
Thus, the algorithm satisfies the properties (1) and (2).

The algorithm is not lock-free, as all threads must increment the \Op{polling} and \Op{idle} counters before any thread can terminate.
If a thread halts before completing both increments, it delays the termination of all other threads.
However, the algorithm does allow threads to leave and join the system at specific points.
Note that arbitrarily suspending threads would break most applications, regardless of the termination detection algorithm.
For example, a thread could be suspended right after deleting an element but before processing it.
To avoid this issue, deleting and processing an element (including possible insertions) would need to occur within a single atomic transaction, which is beyond the scope of this work.

The termination detection algorithm can be adapted for PQ implementations where deletions may fail arbitrarily.
Then, each thread additionally counts its insertions and deletions locally.
Instead of terminating when all threads are idle, they synchronize to sum up all counters.
If the total number of insertions is equal to the total number of deletions, all threads can terminate.
Otherwise, the threads are reset to the working state and continue processing elements.

\section{Experiments}\label{experiments}
The experiments can be divided into three parts.
In the first part, we measure the rank errors and delays exhibited by the basic MultiQueue and compare them to the results of the theoretical analysis.
In the second part, we perform parameter tuning of the MultiQueue, including the practical improvements, to obtain a set of configurations that offer a good trade-off between throughput and quality.
Finally, in the last part, we conduct an extensive comparison of the MultiQueue with its competitors, evaluating performance across various workloads and different hardware architectures.
This includes stress tests and concurrent variants of Dijkstra's algorithm and a branch-and-bound algorithm.
The latter might be of independent interest as parallel algorithms in their own right.

\subsection{Implementation}\label{experiments:implementation}
At its core, the MultiQueue is an array with one entry for each internal PQ.
Each entry contains a mutual exclusion lock, the insertion and deletion buffers, a redundant copy of the smallest element, and a pointer to the PQ itself.
Thus, the array requires a fixed amount of memory in $\BigO(p)$.
Comparing the smallest elements of internal PQs is done by comparing the redundant copies, which does not require locking the PQs.
Entries and the buffers themselves are padded and aligned to \emph{cache lines} to prevent false sharing of neighboring entries.
Thread-specific data, such as stickiness counters, are stored in thread-local storage.

The insertion buffer consists of an array and a size counter.
The deletion buffer is implemented as a ring buffer with elements sorted in ascending order.
Ring buffers support removing the first element in constant time, while insertions at arbitrary positions take at most $n/2$ element moves with $n$ elements.
The PQs themselves are implicit $k$-ary heaps stored in dynamically growing arrays.

Our implementation of the MultiQueue data structure allows the user to easily exchange the PQ implementation and the
memory allocation strategy.
The priority queue implementation we use for our experiments is built on top of \texttt{std::vector} using the default \texttt{std::allocator} from the C\texttt{++} standard library.
Memory is only returned to the system when the PQ is destructed.
Thus, the overall memory consumption is in $\BigO(p+n)$, where $n$ is the maximum previous size of the queue.

While our implementation can handle generic element types, it is designed and optimized for elements of small size that are cheap to copy and compare.

\subsection{Methodology}\label{experiments:methodology}
\paragraph{Stress Tests}
The stress  tests are designed to measure the maximum throughput and quality under high contention.
For compatibility with all tested priority queue implementations, elements are key-value pairs of two unsigned 64-bit integers.
The order of the elements is determined by the key.
The value of each element is a globally unique identifier to ensure that no two elements are equal.
The following stress tests are conducted:
\begin{description}
	\item[Monotonic:]
	      At first, $n$ elements with keys $\{1,\ldots,n\}$ are inserted into the PQ to simulate a busy system (the \emph{pre-fill}).
	      Then, each thread repeatedly performs alternating \Op{delete} and \Op{insert} operations until a fixed number of total iterations is reached.
	      The keys of inserted elements are drawn uniformly at random from the range $[k,k+n]$, where $k$ is the key of the previously deleted element.
	\item[Insert-Delete:]
	      For some number of elements $n$, all threads insert elements with keys drawn uniformly at random from the range $[1,n]$ until $n$ elements are inserted.
	      Then, all threads perform \Op{delete} operations until the PQ is empty.
\end{description}
The monotonic stress test keeps the total number of elements in the PQ constant and roughly mimics applications where the elements monotonically become larger, such as Dijkstra's algorithm and branch-and-bound algorithms.
The rank of newly inserted elements is more evenly distributed compared to the naive approach to draw keys uniformly at random from the range $[1,n]$, where the ranks of newly inserted elements are very low most of the time.
We consider the monotonic stress test as reasonably general to assess the general throughput, rank errors and delays of concurrent PQ implementations.
The insert-delete stress test complements the monotonic test as it measures the maximum throughput of insertions and deletions separately with varying numbers of elements.
The insert-delete workload is reminiscent of heap-sort.

\paragraph{Measuring Rank Errors and Delays}
Measuring the rank errors and delays in real-time imposes the practical problem that we need to know which elements are present in the PQ at the time of each \Op{delete} operation.
We approach this problem as follows.
Each thread logs its operations locally with timestamps using a low-overhead high-resolution clock.\footnote{The resolution of the \texttt{std::chrono::high\_resolution\_clock} from the C\texttt{++} standard library was sufficient on our platforms.
	The timestamp counter (TSC) might have higher resolution but is not guaranteed to be synchronized across cores and/or sockets.}
Insertion-timestamps are recorded immediately \emph{before} each \Op{insert} operation and deletion-timestamps are recorded immediately \emph{after} each \Op{delete} operation.
To compute the rank error and delays, we merge the logs of all threads to one global sequence of operations sorted by timestamp.
We then replay this sequence sequentially and compute the exhibited rank errors and delays for each \Op{delete} operation.
While this approach is not perfectly accurate, we deem it sufficient for our purposes.
Notably, even linearizable PQs can exhibit rank errors due to the time window between the timestamps and the actual linearization points.
We use an augmented B\texttt{+} tree (based on the \verb!tlx::btree_map! from the \texttt{tlx} library~\cite{TLX}) to replay the operations and compute the rank errors and delays.
The B\texttt{+} tree holds the elements in priority order and contains exactly the elements that were in the priority queue at the point in time of the currently replayed operation according to the timestamps.
By augmenting the inner nodes of the B\texttt{+} tree with their size, the rank of deleted elements can be determined in logarithmic time (e.g., \cite[Section~7.5]{sandersSequentialParallelAlgorithms2019}).
Each node $v$ is further augmented with a delay counter $d_v$.
The delay of an element $e$ (leafs in the B\texttt{+} tree) is the sum of the delay counters on the path from the root to $e$.
Since the delay of newly inserted elements is zero, $d_e$ is initially set to the negated sum of all delays on the path from $e$ to the root.
To increase the delay of smaller elements when deleting an element, the delay counters of nodes as close to the root as possible are incremented.
When performing balancing operations on the tree, the delay counters of the manipulated nodes are propagated downward to unchanged subtrees.

\paragraph{Setup}
The MultiQueue and all benchmarks are implemented in C{}\verb!++!17.
Competitors are implemented in C{} and C{}\verb!++!.
All code is compiled with GCC~12.1 using \verb!-O3 -DNDEBUG -march=native!.
We use three different machines for the experimental evaluation:
\begin{description}
	\item[Machine \textsf{AMD}] is equipped with an AMD~EPYC~7702 processor (\texttt{x86\nobreakdash-64} ISA) and \qty{1}{\tera\byte} of DDR4 RAM.
	      The processor hosts 64 cores with 2 hardware threads each.
	      The maximum clock frequency is \qty{3.35}{\giga\hertz}.
	      Each core has \qty{32}{\kilo\byte} of L1 data cache and \qty{1}{\mega\byte} of L2 cache.
	      Four cores share \qty{16}{\mega\byte} of L3 cache.
	      The machine runs Ubuntu~Server~22.04 with Linux~5.4.
	\item[Machine \textsf{Intel-NUMA}] is equipped with four Intel~Xeon~Gold~6138 processors with \qty{188}{\giga\byte} of DDR4 RAM each.
	      Each processor has 20 cores with 2 hardware threads each.
	      The maximum clock frequency is \qty{3.7}{\giga\hertz}.
	      Each core has \qty{32}{\kilo\byte} of L1 data cache and \qty{1}{\mega\byte} of L2 cache.
	      Each NUMA node shares \qty{27.5}{\mega\byte} of L3 cache.
	      The machine runs Ubuntu~Server~22.04 with Linux~5.4.
	\item[Machine \textsf{ARM}] is equipped with a Neoverse\nobreakdash-N1 processor (\texttt{AArch64} ISA) and \qty{256}{\giga\byte} of DDR4 RAM.
	      The processor hosts 80 cores.
	      The maximum clock frequency is \qty{3}{\giga\hertz}.
	      Each core has \qty{64}{\kilo\byte} of L1 data cache and \qty{1}{\mega\byte} of L2 cache.
	      The machine runs Ubuntu~Server~22.04 with the Linux~5.15 kernel.
\end{description}
We conduct the parameter tuning and all comparison experiments on machine \textsf{AMD}, as it is compatible with all implementations and offers the most homogeneous scalability behaviour.
Additionally, the stress tests are performed on the \textsf{Intel-NUMA} and \textsf{ARM} machines.
The throughput of the monotonic stress test is measured over $2^{24}$ iterations per thread, or until a timeout of $\qty{5}{s}$ is reached.
Whenever possible, sufficient memory is pre-allocated to prevent memory allocation during the benchmarks.
Each stress test is repeated \num{5} times and the mean is reported.
Execution threads are pinned to distinct hardware threads, the clock frequency is fixed and frequency boosting is disabled.
We use the \texttt{numactl} tool on the \textsf{Intel-NUMA} machine to distribute the data evenly across the participating NUMA nodes.
Cache misses are measured on the \textsf{AMD} machine programmatically with the \emph{L2 Cache Miss from Data Cache Miss} performance counter\footnote{\url{https://www.amd.com/content/dam/amd/en/documents/processor-tech-docs/programmer-references/55803-ppr-family-17h-model-31h-b0-processors.pdf}} using the \texttt{PAPI}\footnote{\url{https://icl.utk.edu/papi/}} library.

\subsection{Quality of the basic MultiQueue}\label{experiments:quality}
In this section we want to measure the quality characteristics of the MultiQueue design itself and compare it to the results from the theoretical analysis.
In the following, we refer to the analysis given in Section~\ref{analysis:quality} as the estimation and the results from \citeauthor{walzerSimpleExactAnalysis2024} \cite{walzerSimpleExactAnalysis2024} as the prediction.
Both analyses assume that no locking attempt fails and comparisons are never stale.
To avoid such issues, we execute the tests sequentially.
Therefore, we parametrize the MultiQueue by the total number of PQs rather than of the queue factor in this section.
Preliminary experiments with up to \num{128} threads showed that the effects of concurrency have a negligible impact on the quality.

\begin{figure}
	\includegraphics{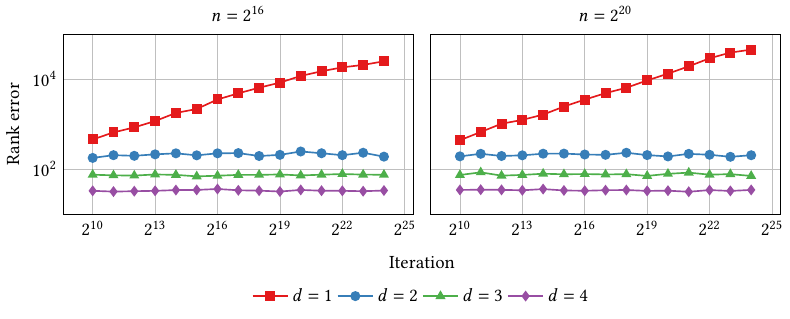}
	\caption{
		Development of the rank error during the monotonic stress test for \num{256} PQs with different pre-fills $n$.
		Each line shows a different number of candidate PQs $d$, and each data point shows the mean of the previous \num{1024} iterations.
	}
	\label{fig:pop-pqs}
\end{figure}
Figure~\ref{fig:pop-pqs} shows the measured quality throughout the monotonic stress test with different numbers of candidate PQs $d$ for the \Op{delete} operation.
When deleting from a randomly chosen PQ ($d=1$), the rank error increases not only with the number of iterations but also with the number of elements in the PQ.
In contrast, considering just two PQs ($d=2$) stabilizes the rank error and makes it independent of the pre-fill.
This result validates the fundamental design of the MultiQueue.
As expected, the rank error further decreases for higher $d$.
We set $d=2$ for the rest of the experiments as it strikes a good balance between quality and cache efficiency.

\begin{figure}
	\includegraphics{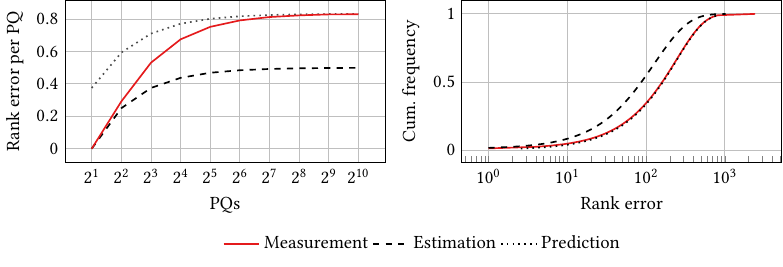}
	\caption{%
	Rank errors during the monotonic stress test with a pre-fill of $n=2^{20}$.
	The left plot shows the mean rank error for different numbers of PQs.
	The right plot shows the cumulative frequency of rank errors less than or equal to the x-axis value with \num{256} PQs.
	For reference, our estimation from Section~\ref{analysis:quality} and the results from \citeauthor{walzerSimpleExactAnalysis2024} \cite{walzerSimpleExactAnalysis2024} are also shown.
	}
	\label{fig:quality_theory}
\end{figure}
Figure~\ref{fig:quality_theory} shows the mean rank error for a varying number of PQs, and the distribution of rank errors for \num{256} PQs.
The mean rank error is indeed linear in the number of PQs and the factor converges almost exactly to the predicted value of $\tfrac{5}{6}$, which is slightly worse than our estimation of $\tfrac{1}{2}$.
The reason for the discrepancy for smaller numbers of PQs is likely due to an optimization in the MultiQueue that ensures that it picks two distinct PQs for deletion, whereas the prediction assumes that both PQs are chosen independently.
For larger numbers of PQs, the optimization becomes less significant.
The shape of the rank error distribution also closely matches the prediction, which is slightly worse than our estimation.

\subsection{Parameter tuning}
In this section, we explore the parameter space of the MultiQueue systematically and identify a set of good configurations.
We generally optimize for good scalability and performance with many threads.
We first tune the buffers and the $k$-ary heap, and then find combinations of stickiness period and queue factors that yield good throughput--quality trade-offs.
We evaluate buffer sizes of $C\in\{0,4,16,64,256,1024\}$ and heap arities of $k\in\{2,4,8,16\}$.
To reduce complexity, we only consider the same buffer size for both the insertion and deletion buffer.
\begin{figure}
	\includegraphics{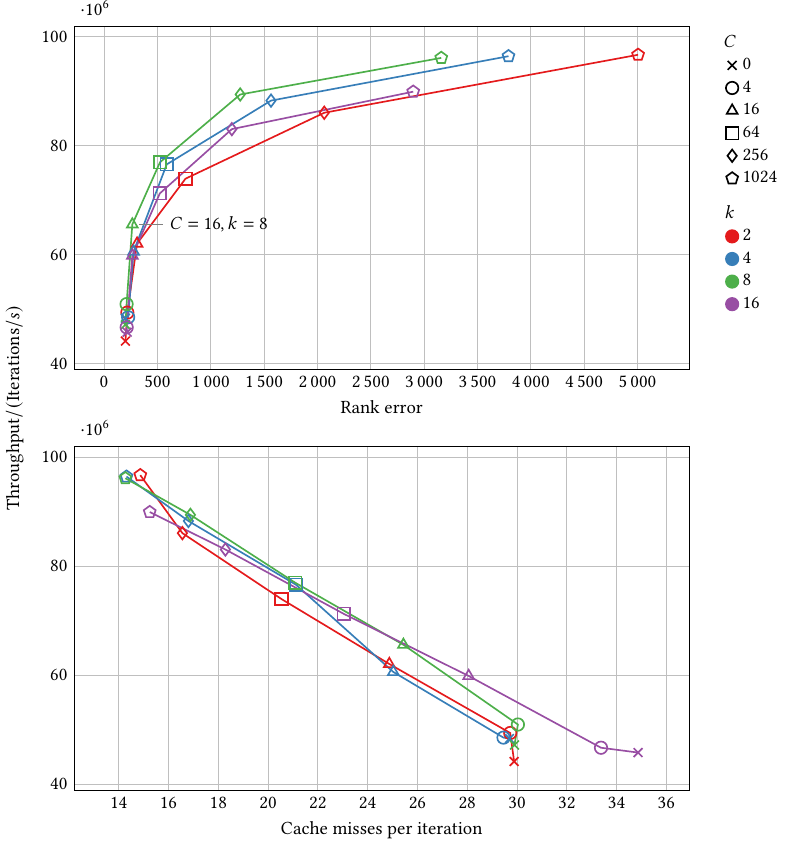}
	\caption{
	Throughput versus rank error and L2 cache misses for varying buffer sizes $C$ and heap arities $k$ with \num{128} threads and $n=2^{23}$ elements pre-fill.
	The buffer size is the same for the insertion buffer and the deletion buffer.
	The annotated data point corresponds to the final configuration with $C=16$ and $k=8$.
	}
	\label{fig:parameter-tuning}
\end{figure}
Figure~\ref{fig:parameter-tuning} shows a clear correlation between buffer size, throughput, and cache misses.
The correlation is in line with our expectations that the performance of the MultiQueue is limited by cache misses, and confirms that buffering is an effective mitigation measure.
While larger buffers increase the overall throughput due to higher cache locality, they also increase the runtime of operations that need to flush or refill buffers.
Higher variance in the runtime of operations leads to deteriorating quality, which can be observed with buffer sizes larger than $C=16$.
Buffer sizes larger than $C=256$ yield only minor improvements in throughput while worsening the quality significantly.
The heap arity only has a minor impact on the throughput and quality, but $k=8$ seems to be the best choice overall.
Our final configuration with $C=16$ and $k=8$ yields the highest throughput without sacrificing quality.

The queue factor $c$ and the stickiness period $s$ offer a trade-off between throughput and quality.
Intuitively, the higher the number of PQs and the longer the stickiness period, the higher the throughput and the rank errors.
\begin{figure}
	\includegraphics{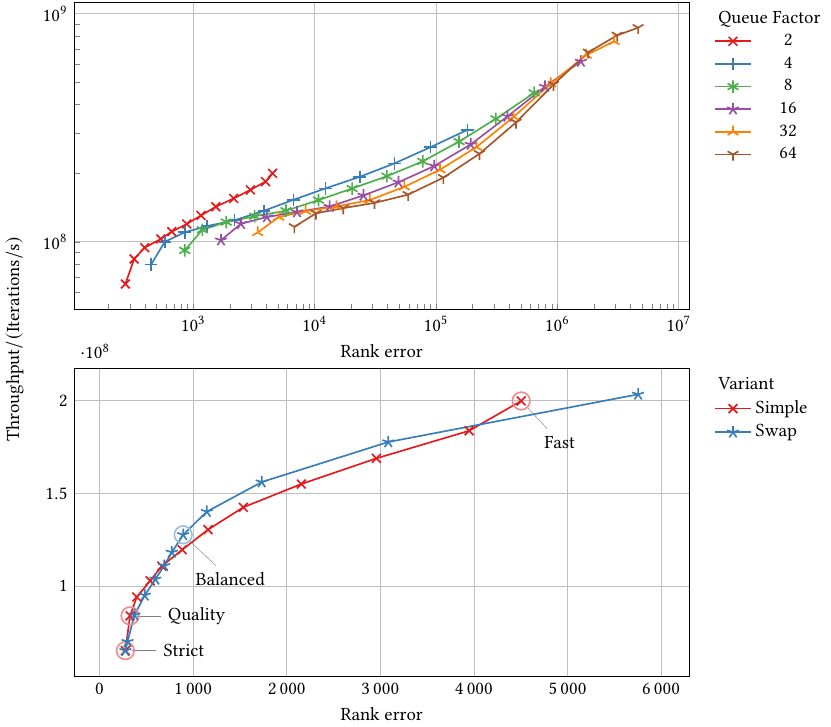}
	\caption{The mean rank error and throughput with stickiness for different queue factors with \num{128} threads and a pre-fill of $n=2^{23}$.
	The points are connected, starting with no stickiness, followed by increasing stickiness periods $2^2,2^3,\ldots,2^{12}$, for better readability.
	The upper plot shows the impact of stickiness (simple variant) for different queue factors.
	The lower plot compares the simple and swap stickiness variants for $c=2$.
	}
	\label{fig:throughput-quality-c-vs-s}
\end{figure}
Figure~\ref{fig:throughput-quality-c-vs-s} confirms this intuition and shows that the quality deteriorates more quickly with increasing queue factors than with higher stickiness.
Consequently, configurations with $c=2$ dominate configurations with higher $c$ and comparable throughput.
The speedup from increasing the stickiness period eventually stagnates, presumably because the cache misses from switching PQs become insignificant.
For higher throughput, the queue factor has to be increased at the cost of significantly worse quality.
With higher queue factors, each PQ contains fewer elements, so the operations are faster and longer stickiness periods are possible before the throughput stagnates.
Notably, for $c=2$ and the simple stickiness variant, not only the throughput, but also the quality stagnates.
A possible explanation is that switching PQs can trigger other threads to switch PQs as well, before their stickiness period is over.
This cascading effect is less pronounced with higher queue factors, since it becomes more likely to switch to an unused PQ.
The swap stickiness variant also mitigates the cascading effect and yields slightly higher throughput than the simple variant with comparable quality.

\begin{table}
	\caption{A selection of Pareto-optimal configurations.}
	\label{tab:config}
	\begin{tabular}{l l S[table-format=2.0] S[table-format=4.0] S[table-format=3.0] S[table-format=7.0]}
		Name     & {Stickiness} & {$c$} & {$s$}     & Throughput                               & {Rank error} \\
		         &              &       &           & {$10^6\text{Iterations}/\unit{\second}$} &              \\
		\hline
		Strict   & No           & 2     & \text{--} & 65                                       & 276          \\
		Quality  & Simple       & 2     & 4         & 84                                       & 326          \\
		Balanced & Swap         & 2     & 256       & 128                                      & 894          \\
		Fast     & Simple       & 2     & 4096      & 200                                      & 4505         \\
		\hline
		         & Simple       & 4     & 4096      & 310                                      & 181811       \\
		         & Simple       & 16    & 4096      & 619                                      & 1563132      \\
		         & Simple       & 64    & 4096      & 870                                      & 4623838      \\
	\end{tabular}
\end{table}
Table~\ref{tab:config} shows configurations that yield interesting trade-offs between average rank error and throughput.
For comparison with the competitors, we only consider configurations with $c=2$, as other Pareto-optimal configurations with $c>2$ exhibit significantly worse quality.
The sharp increase in rank error for $c>2$ is likely due to the longer time required for all PQs to be selected for deletion by some thread again, reminiscent of the coupon collector's problem.
Workloads in which not all threads frequently perform deletions (such as producer-consumer scenarios) may experience higher rank errors than in the monotonic stress test for the same configurations.

\subsection{Comparison}
We now compare the final configurations of the MultiQueue to state-of-the-art concurrent priority queues found in the literature.
Besides relaxed PQs, we include linearizable PQs in order to assess the trade-off between quality and throughput in real-world applications.
In particular, we include the linearizable PQ by \citeauthor{lindenSkiplistBasedConcurrentPriority2013}~\cite{lindenSkiplistBasedConcurrentPriority2013}, as it is a popular point of reference in the literature (e.g., \cite{braginskyCBPQHighPerformance2016,rukundoTSLQueueEfficientLockFree2021,sagonasContentionAvoidingConcurrent2017,wimmerLockfreeKLSMRelaxed2015,alistarhSprayListScalableRelaxed2015})
While there is more recent work on linearizable PQs (e.g., \cite{braginskyCBPQHighPerformance2016,rukundoTSLQueueEfficientLockFree2021}) that outperforms the Linden PQ empirically by up to a factor of $5$, we expect all linearizable PQs to have significantly worse scalability and throughput (with high thread counts) than relaxed PQs.
As a baseline for the stress tests, we also include the buffered $k$-ary heap used by the MultiQueue, guarding each access with a mutual exclusion lock.
The following list gives an overview of the competitors.
Detailed descriptions of most competitors can be found in Section~\ref{related-work}.
\begin{description}
	\item [MQ strict, quality, balanced, fast]
	      MultiQueue configurations from Table~\ref{tab:config}.
	\item [$k$-LSM 4, 256, 1024, 4096] The $k$-LSM by \citeauthor{wimmerLockfreeKLSMRelaxed2015}~\cite{wimmerLockfreeKLSMRelaxed2015} with $k=4$, $k=256$, $k=1024$ and $k=4096$, respectively.
	\item [CA-PQ] The contention-avoiding priority queue by \citeauthor{sagonasContentionAvoidingConcurrent2017}~\cite{sagonasContentionAvoidingConcurrent2017}.
	\item [Spraylist] The Spraylist by \citeauthor{alistarhSprayListScalableRelaxed2015}~\cite{alistarhSprayListScalableRelaxed2015} in the default configuration.
	\item [SMQ]
	      Standalone adaption of the stealing MultiQueue implementation\footnote{\url{https://github.com/npostnikova/mq-based-schedulers}} by \citeauthor{postnikovaMultiqueuesCanBe2022}~\cite{postnikovaMultiqueuesCanBe2022}.
	      We use a stealing probability of $\beta=\tfrac{1}{8}$ and a batch size of $b=64$ as suggested by the authors.
	\item [Linden] Linearizable PQ by \citeauthor{lindenSkiplistBasedConcurrentPriority2013}~\cite{lindenSkiplistBasedConcurrentPriority2013}.
	\item [TBB PQ]
	      Linearizable Concurrent priority queue from Intel® oneAPI Threading Building Blocks.\footnote{The concurrent priority queue from Intel® oneAPI Threading Building Blocks.}
	\item [Locked Heap]
	      Buffered $k$-ary heap used by the MultiQueue with a mutual exclusion lock for each operation.
	      This competitor is included to provide a linearizable baseline for the stress tests.
\end{description}
The implementations of the CA-PQ, the Spraylist and the Linden PQ are taken from the $k$-LSM repository\footnote{\url{https://github.com/klsmpq/klsm}}.
Unfortunately, the implementations of the $k$-LSM, the CA-PQ, the Spraylist and the Linden PQ are not compatible with the ARM \texttt{AAarch64} ISA; they either fail to compile or behave incorrectly due to the more relaxed memory ordering compared to \texttt{x64\nobreakdash-64}.
The \Op{delete} operations of the $k$-LSM and Spraylist implementations sometimes fail unexpectedly.
Therefore, we use the adapted termination detection algorithm described in Section~\ref{termination} for these competitors when necessary.\footnote{We do not see fundamental reasons why these PQs could not be implemented in a way that satisfies the assumption.}

\subsubsection{Stress tests}
\begin{figure}
	\includegraphics{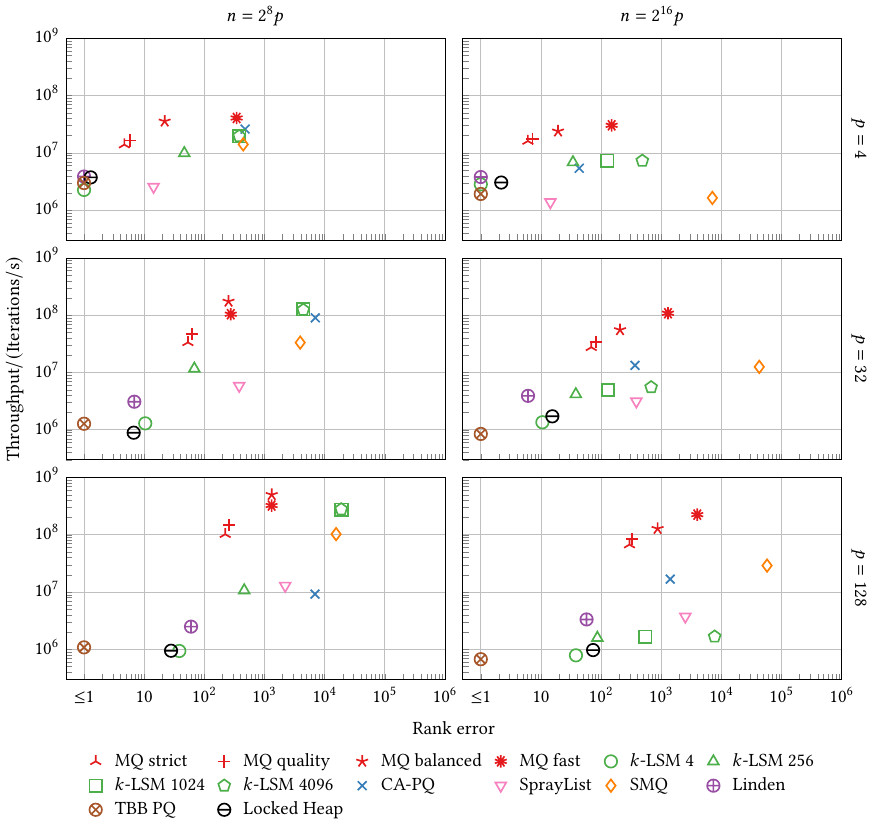}
	\caption{%
		Mean rank error versus throughput for the monotonic stress test with different pre-fills $n$ on machine \textsf{AMD}.
		Each row of plots corresponds to a different thread count ($p$) and each column to a different pre-fill ($n$).
		The label ``${\leq}1$'' indicates that the rank error is less than or equal to $1$.
	}
	\label{fig:comparison-throughput-quality}
\end{figure}
Figure~\ref{fig:comparison-throughput-quality} plots the quality against the throughput for all competitors on the monotonic stress test with two different pre-fills on machine \textsf{AMD}.
The full results with a larger spectrum of thread counts and pre-fills are given in Appendix~\ref{appendix:throughput_quality}.
The MultiQueue configurations dominate the Pareto-front except for very small rank errors with all tested pre-fills.
When comparing with competitors that achieve similar rank errors on \num{128} threads, the advantages range from $4.5\times$ to $25\times$ in throughput.
With large pre-fills, the MultiQueue achieves significantly higher throughput than all other competitors.
The $k$-LSM~1024 and the $k$-LSM~4096 achieve similar throughput as the fast MultiQueue with \num{128} threads for the smallest pre-fill but exhibit an order of magnitude higher rank errors.
On large pre-fills, the $k$-LSM fails to scale and is slower than the MultiQueue by almost two orders of magnitude with \num{128} threads.
The CA-PQ achieves competitive throughput with the MultiQueue but fails to scale beyond \num{32} threads on small pre-fills.
On large pre-fills, the CA-PQ exhibits similar quality to the balanced MultiQueue but is \num{5} times slower.
The SMQ offers competitive throughput to the MultiQueue on small pre-fills, but the quality is at least an order of magnitude worse in most cases.
Moreover, the quality of the SMQ deteriorates with the pre-fill, and it exhibits the highest rank errors of all competitors on large pre-fills.
The Spraylist is robust to the pre-fill and exhibits rank errors very similar to the balanced MultiQueue but is an order of magnitude slower.
Unsurprisingly, the linearizable PQs and the $k$-LSM~4 yield very low rank errors but do not offer competitive throughput, especially with higher thread counts.

\begin{figure}
	\includegraphics{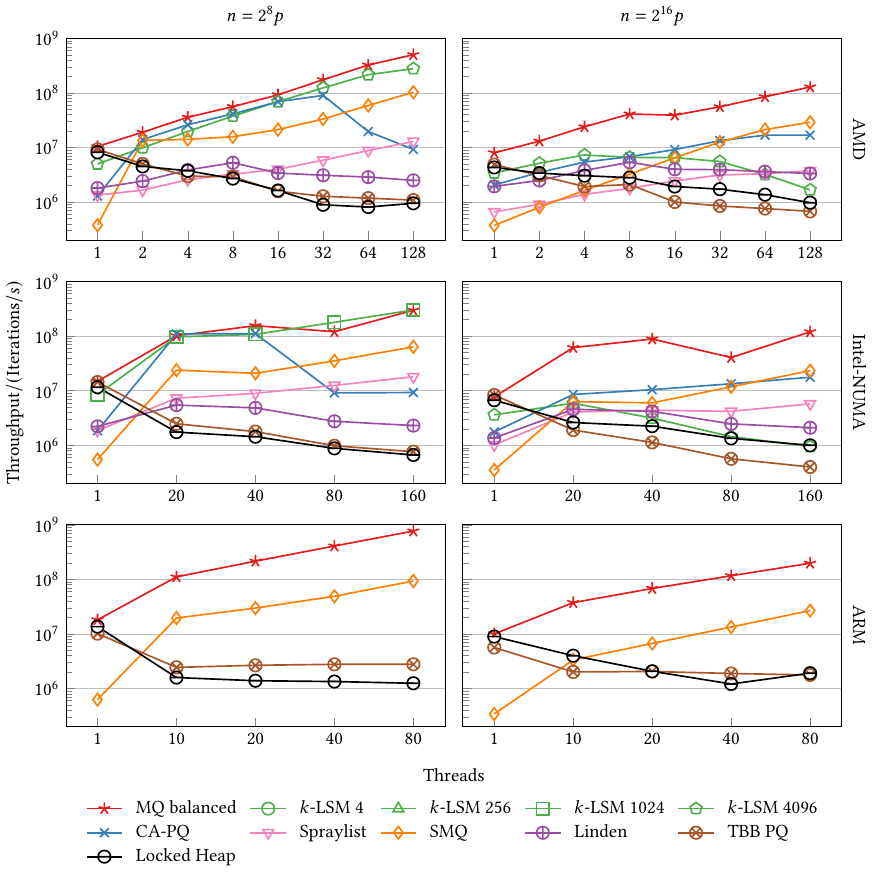}
	\caption{
		Weak scaling experiment showing the throughput for the monotonic stress test versus the number of threads on different machines.
		Each row corresponds to a different machine and each column to a different pre-fill.
		Only the $k$-LSM variant with the highest throughput is shown for better readability.
		Missing plots on the \textsf{ARM} machine are due to incompatibility with the \texttt{AArch64} ISA.
	}
	\label{fig:comparison-throughput-scaling}
\end{figure}
Figure~\ref{fig:comparison-throughput-scaling} shows the corresponding throughput scaling behaviour of all competitors on all three machines using the same pre-fills as in Figure~\ref{fig:comparison-throughput-quality}.
For better readability, only the balanced MultiQueue configuration and the best (highest maximum throughput) $k$-LSM variant are shown.
The full results are given in Appendix~\ref{appendix:scaling}.
The MultiQueue variants generally scale well, even when using hardware multithreading.
The performance dips slightly when going upwards in the architectural hierarchy.
This effect is most pronounced on the Intel-NUMA machine when going from two to four sockets.
However, this is a price one is often willing to pay in order to be able to globally coordinate multiple threads by priority driven scheduling of work---the overall application may still profit from being able to use more threads.

As a close relative to the MultiQueue, the SMQ has similar scaling behavior, albeit with lower absolute performance.
The CA-PQ scales up to a moderate number of threads but then experiences sharp drops in performance, likely when it has to fallback to the global skip list.
We view it as likely that adapting some tuning parameters might be able to remedy this problem, albeit at the price of even higher rank errors.
The more relaxed variants of the $k$-LSM scale well for small queues, but the scalability is poor in all other configurations.
We attribute the particularly bad performance for large queues to the bottleneck introduced by the expensive merging within LSM trees.
As expected, the linearizable queues do not scale at all.
Throughput even \emph{decreases} by an order of magnitude when going to the maximum number of threads, and can be up to \num{300} times slower than the fastest relaxed queues on the \textsf{AMD} machine.
This underlines the utility of relaxed queues.

\begin{figure}
	\includegraphics{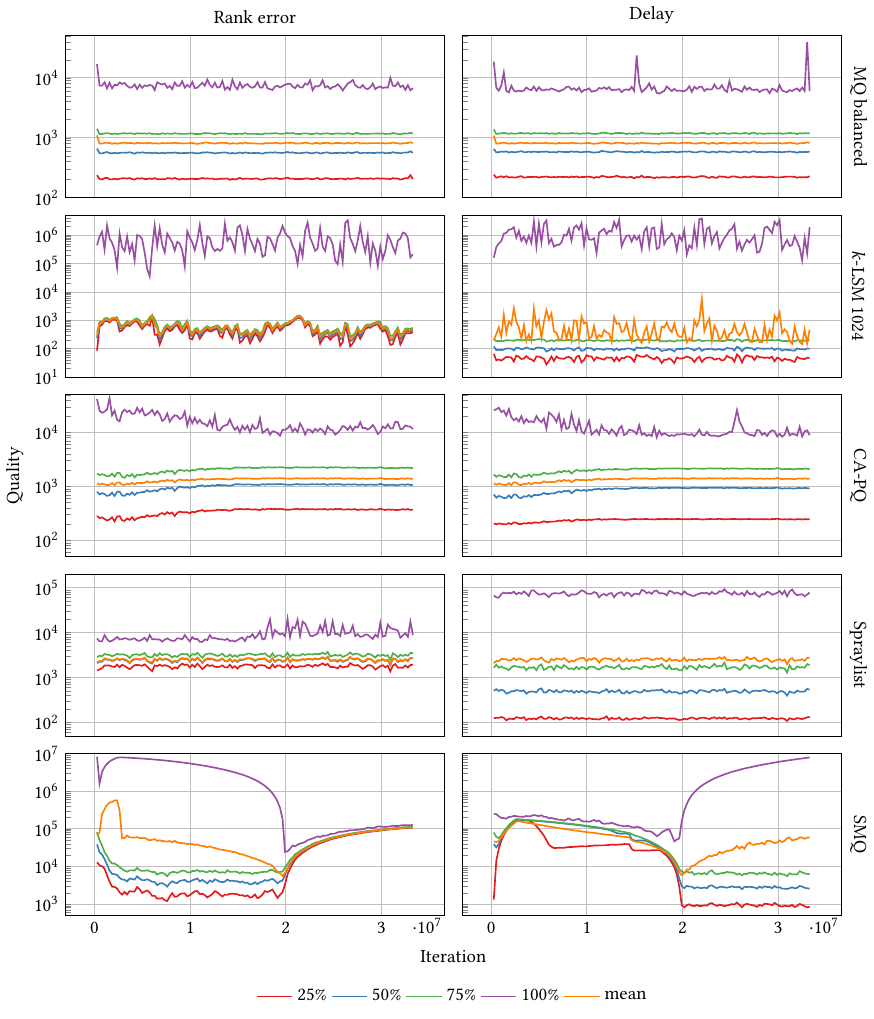}
	\caption{
	Development of the rank errors and delays of the monotonic stress test with a pre-fill of $n=2^{16}p$ and $128$ threads on machine \textsf{AMD}.
	The iterations are grouped into bins of size $2^{18}$, and for each bin the $25\%$, $50\%$, $75\%$ and $100\%$ quantiles as well as the mean are plotted.
	Each row corresponds to one competitor and the two columns show the rank errors and delays, respectively.
	}
	\label{fig:comparison-quality-binned}
\end{figure}
Examining the rank errors and delays in more detail, Figure~\ref{fig:comparison-quality-binned} shows their development over time for a selection of the relaxed PQs.
The MultiQueue exhibits very stable rank errors and delays over time, and both metrics behave very similarly.
The $25\%$ and $75\%$ quantiles of the rank errors and delays are close together, indicating that most rank errors and delays are within a small range.
The largest measured rank errors and delays are roughly one order of magnitude larger than the mean.
The rank errors and delays fluctuate much more for the $k$-LSM than for the other queues.
These fluctuations reach an order of magnitude for the maxima and also exceed the claimed rank error bounds by an order of magnitude.
This is likely due to occasional expensive sequential merge operations of large parts of the log-structured merge tree during which no elements can be deleted from the merged lists.
While the rank error quantiles are very close together, the delay quantiles are more spread out.
The CA-PQ behaves similarly to the MultiQueue but exhibits a slight increase in mean rank errors and delays over time, while the maximum rank errors and delays decrease.
The Spraylist exhibits a very small range of rank errors but delays are widely spread.
The SMQ shows very inconsistent behavior over time, which is unexpected for such a homogeneous benchmark, and given the consistent behavior of the other competitors.
Due to these inconsistencies, the overall mean rank errors and delays vary significantly with the number of iterations for the SMQ.

\begin{figure}
	\includegraphics{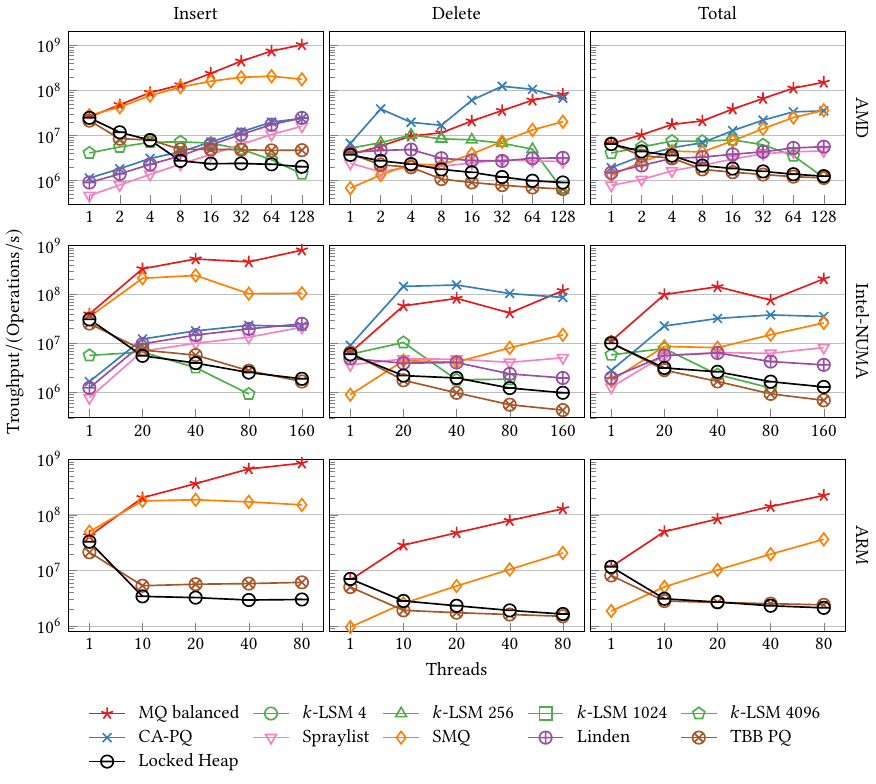}
	\caption{
		Throughput comparison for the insert-delete stress test with $2^{20}$ elements per thread on different machines.
		Only the $k$-LSM variant with the highest throughput is shown for better readability.
		Data points for the $k$-LSM variants on machine \textsf{Intel-NUMA} for \num{160} threads are missing due to crashes.
		Missing lines in the \textsf{ARM} plot are due to incompatibility with the ARM \texttt{AArch64} ISA.
	}
	\label{fig:comparison-push-pop}
\end{figure}
Figure~\ref{fig:comparison-push-pop} shows the throughput as the average time per operation for the insert-delete stress test on all machines.
Again, we only show the balanced MultiQueue configuration and the best $k$-LSM variant for better readability.
The full results are given in Appendix~\ref{appendix:insdel}.
Most competitors, with the exception of the CA-PQ, finish the insertions faster than the deletions.
Only the SMQ can compete with the insertion times of the MultiQueue; however, it scales worse with the number of threads.
All other competitors are at least an order of magnitude slower than the MultiQueue and SMQ for high thread counts.
The CA-PQ generally has the fastest deletions for low thread counts but is slower than the MultiQueue for high thread counts.

\subsubsection{Single-Source Shortest-Path}\label{experiments:SSSP}
The single-source shortest-path (SSSP) problem is a fundamental and widely known graph problems.
Given a weighted graph $G$ and a source node $s$ in $G$, the goal is to find the shortest path from $s$ to all other nodes in $G$.
The most famous algorithm for the SSSP problem with non-negative edge weights is probably Dijkstra's algorithm \cite{dijkstraNoteTwoProblems2022}.
Hence, a natural parallel algorithm is a modified version of Dijkstra's algorithm using a concurrent PQ.
\footnote{
	Since we use the algorithm as a benchmark for PQs, we do not prioritize the competitiveness of the algorithm itself.
	Several specialized parallel algorithms for the SSSP exist (e.g., $\Delta$-stepping \cite{meyerDsteppingParallelizableShortest2003}).
}
The algorithm maintains an array of tentative distances to each node and a PQ that stores nodes ordered by their tentative distances.
Initially, the tentative distance is $0$ for $s$ and $\infty$ for all other nodes, with $s$ being the only node in the PQ.
Then, all threads repeatedly process nodes from the PQ.
If a node's tentative distance remains unchanged since its insertion, the node is \emph{scanned}; otherwise, it is discarded.
Scanning involves checking each outgoing edge to determine whether it yields a shorter tentative distance to a neighboring node.
If so, the neighbor's tentative distance is updated atomically in the array, and the neighbor is (re)inserted into the PQ.
Note that nodes are reinserted into the PQ with improved tentative distances instead of being updated in-place since most implementations of relaxed PQs do not support these updates.
The algorithm terminates when the PQ is empty and no thread is currently scanning a node.
Since the total number of operations is unknown in advance, we employ the termination detection algorithm described in Section~\ref{termination}.

Nodes can be deleted and scanned before all nodes with shorter tentative distances are fully scanned, meaning shorter paths to a node may still be found even after it has been scanned.
This results in redundant scans, making the algorithm \emph{label-correcting}, even with linearizable PQs.
In contrast, classical Dijkstra's algorithm is \emph{label-setting}---each node is only scanned once and its distance does not change afterward.
The number of redundant scans depends on the quality of the PQ, but in the worst case, a node might be scanned an exponential number of times relative to the number of nodes in the graph \cite{shierPropertiesLabelingMethods1981}.

\begin{table}
	\caption{Graphs used for the SSSP problem.}
	\label{tab:graphs}
	\begin{tabular}{l | S[table-format=2.1] S[table-format=3.1] S[table-format=6.0]}
		Name              & {$\text{Nodes}/ 10^6$} & {$\text{Edges}/ 10^6$} & {Max. degree} \\
		\hline
		NY                & 0.3                    & 0.7                    & 8             \\
		CAL               & 1.9                    & 4.7                    & 8             \\
		CTR               & 14.1                   & 34.3                   & 9             \\
		GER               & 20.7                   & 41.8                   & 9             \\
		USA               & 23.9                   & 58.3                   & 9             \\
		\hline
		$\text{RHG}_{20}$ & 1                      & 10.3                   & 90391         \\
		$\text{RHG}_{22}$ & 4.2                    & 41.3                   & 940912        \\
		$\text{RHG}_{24}$ & 16.8                   & 159.7                  & 594914        \\
	\end{tabular}
\end{table}
We evaluate the SSSP problem on road networks and random hyperbolic graphs (RHGs), and report the solving time for different thread counts as well as the number of scanned nodes.
The edge weights in the road networks correspond to an estimated travel time, and the weights in the hyperbolic graphs correspond to the integer approximation of the hyperbolic distance.
All road networks were obtained from the $9$th DIMACS implementation challenge\footnote{\url{http://www.diag.uniroma1.it/challenge9/download.shtml}}, except for the GER graph\footnote{\url{https://i11www.iti.kit.edu/resources/roadgraphs.php}}.
We used KaGen~\cite{funkeCommunicationfreeMassivelyDistributed2019} to generate hyperbolic graphs with $2^{20}$, $2^{22}$ and $2^{24}$ nodes, a configured average degree of $16$ and a power law exponent of $\gamma=2.3$.
An overview of the graphs is given in Table~\ref{tab:graphs}.

\begin{figure}
	\includegraphics{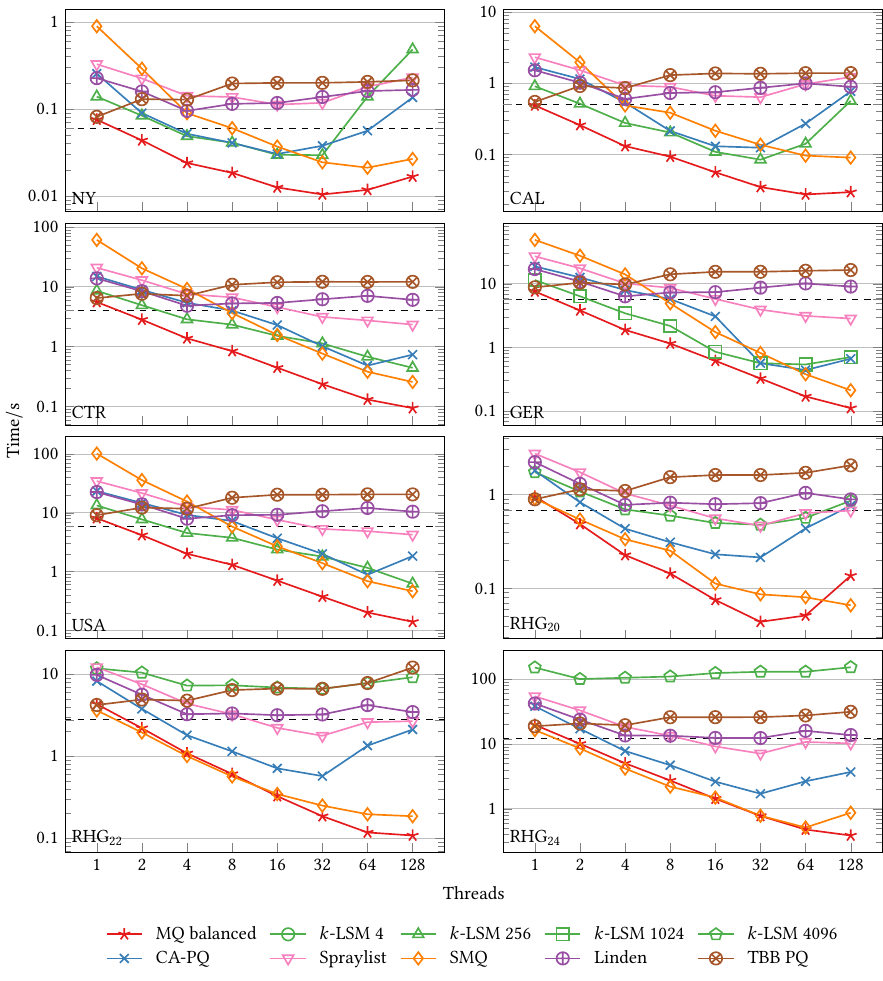}
	\caption{
		Time to solve the SSSP problem with different numbers of threads on machine \textsf{AMD}.
		The black dashed line is the solving time of the sequential algorithm.
		Only the $k$-LSM variant with the fastest solving time is shown for better readability.
	}
	\label{fig:compare-dijkstra-time}
\end{figure}
\begin{figure}
	\includegraphics{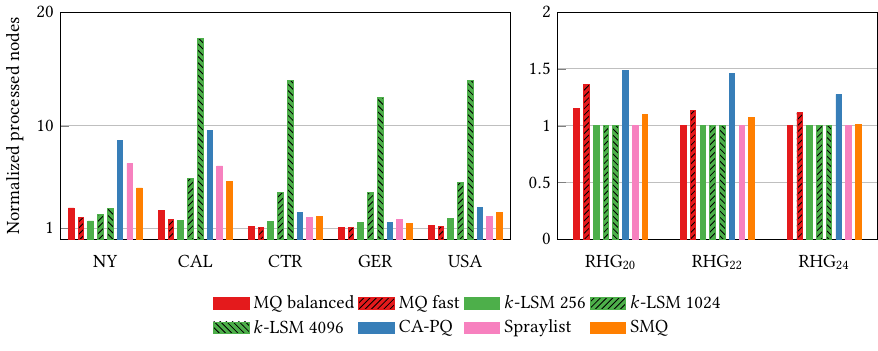}
	\caption{
		Normalized processed nodes with \num{128} threads with respect to the sequential SSSP algorithm.
		The normalization is with respect to the number of processed nodes by the sequential algorithm.
		For clarity, only the best performing priority queues are shown.
		The $k$-LSM~256 and 1024 do not finish within the time limit of \qty{300}{\second} on graph $\text{RHG}_{24}$.
	}
	\label{fig:compare-dijkstra-processed-nodes}
\end{figure}
Figure~\ref{fig:compare-dijkstra-time} shows the time to solve the SSSP problem on these graphs with varying numbers of threads.
Again, we only show the balanced MultiQueue configuration and the fastest $k$-LSM variant for better readability.
The full results are given in Appendix~\ref{appendix:sssp}.
The MultiQueue and the SMQ generally scale well with the number of threads.
The MultiQueue yields faster solving times than all other competitors on all graphs, achieving speedups of \num{50} with \num{128} threads over the sequential algorithm.
These results indicate that it might be interesting to investigate a relaxed version of Dijkstra's algorithm as a competitive parallel SSSP solver.
The SMQ offers similar performance to the MultiQueue on the hyperbolic graphs but is slightly slower on the road networks.
While the $k$-LSM has good speedups on the larger road networks, it does not scale at all on the random hyperbolic graphs.
This is consistent with the observations from our stress tests, where we observed that the $k$-LSM performs best when the queues are small---road networks have high diameter and thus small average search frontiers while the search frontier in RHGs quickly expands to encompass a large part of the graph.
Notably, the $k$-LSM~1024 performs better than the more relaxed variants on the road networks.
On the $\text{RHG}_{20}$ graph, the SMQ is faster than the MultiQueue for $128$ threads, since the MultiQueue fails to scale beyond $32$ threads.
However, the strict and quality variants of the MultiQueue still outperform the SMQ (refer to Appendix~\ref{appendix:sssp}).
On the RHG graphs, the fast MultiQueue does not scale beyond \num{16} threads and becomes slower than the other MultiQueue configurations.

To demonstrate the effects of relaxation on performance, Figure~\ref{fig:compare-dijkstra-processed-nodes} shows the excess work done by the best performing PQs.
Besides the much slower $k$-LSM 256, the MultiQueue variants process the fewest nodes on the road networks, which is consistent with the solving times.
The $k$-LSM~4096 processes the most nodes on road networks (except for the NY graph), up to $18\times$ more than the sequential algorithm and three times as many as the $k$-LSM~1024.
This explains the better performance of the $k$-LSM variants with smaller $k$.
The amount of excess work done on the hyperbolic graphs is generally much lower than on the road networks.
A possible reason is that the much lower diameter of hyperbolic graphs leads to much larger average queue size.
In comparison, the rank errors are very small.
In other words, the shortest path algorithm has plenty of nodes in the queue to choose from, and the nearest of those nodes are likely to have already reached their final shortest path distance.
While the CA-PQ exhibits the most overhead, the increased overhead of the fast MultiQueue compared to the balanced MultiQueue seems to be responsible for the worse performance on the RHG graphs.

In summary, relaxation is beneficial for the concurrent version of Dijkstra's algorithm, but the amount of relaxation has to be chosen carefully to avoid excessive node scans.

\subsubsection{Knapsack}\label{experiments:knapsack}
Lastly, we evaluate the performance of the competitors on a branch-and-bound algorithm for the knapsack problem.
The PQ stores partial solutions, each containing the current value, weight, and index up to which items have been considered.
The lower bound of a partial solution is given by greedily adding items to the knapsack in order of value density until the next item no longer fits or no items remain.
To speed up this step, we use a \emph{finger search}, starting the search for the next item at the index of the last added item.
The upper bound is computed by adding the fractional portion of the next item to completely fill the knapsack.
These bounds are used to prune solutions that cannot be optimal.
The partial solutions in the PQ are ordered in decreasing order of the upper bound.
Similar to the SSSP problem, the algorithm terminates when the PQ is empty, requiring the termination detection mechanism described in Section~\ref{termination}.

We generate instances following a standard approach \cite[Section 2.10]{martelloKnapsackProblemsAlgorithms1990},
adapted to integer weights and values; see also \cite{sandersLastverteilungsalgorithmenFuerParallele1997}.
For each item, the weight $w$ is drawn uniformly at random from $[1,W]$, and the value $v$ is given by $v=w+k$, where $k$ is drawn uniformly at random from $[V, fV]$ for $W\in\{1000,10000\},V=\tfrac{W}{10},f\in\{1.1,1.15,1.2\}$.
The target capacity is $50\%$ or $90\%$ of the total weight of all items.
These parameters are chosen to generate ``interesting'' instances that are nontrivial to solve yet allow for significant pruning opportunities.
For each parameter combination, we generate \num{100} instances with \num{1000} items.
From these, we select instances that the sequential algorithm solves within \num{1}--\qty{15}{\second}, resulting in a final set of \num{147} instances.
Each competitor is given a time limit of \qty{60}{\second}, leading to a total thread-time with \num{128} threads that is $512\times$ that of the sequential algorithm.

\begin{figure}
	\includegraphics{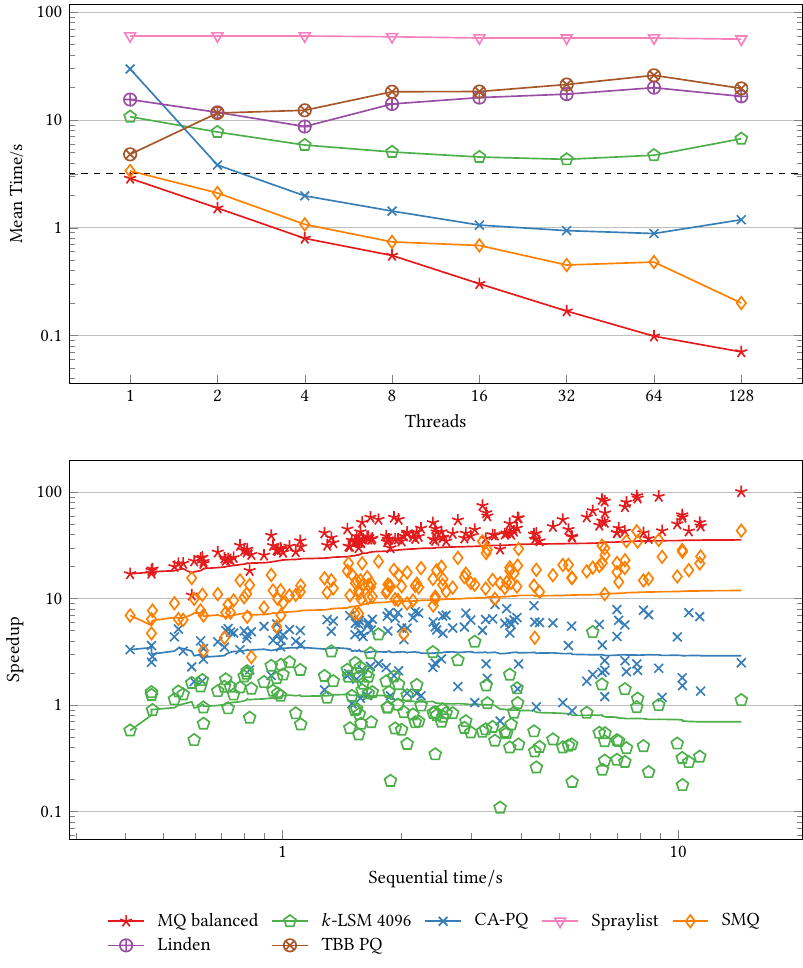}
	\caption{
		Top: Arithmetic mean of times to solve each instance of the knapsack problem with different thread counts on machine \textsf{AMD}.
		If an instance is not solved within the time limit of \qty{60}{\second}, the solving time is set to the time limit.
		The dashed black line indicates the sequential solving time.
		Bottom: Speedup with \num{128} threads for each instance, ordered by the sequential solving time.
		The lines are the cumulative harmonic mean (from left to right) of the speedup.
		Only the best performing relaxed priority queues are shown for clarity.
	}
	\label{fig:compare-knapsack}
\end{figure}
\begin{table}
	\caption{
		Processed nodes and speedup for the knapsack problem with \num{128} threads.
		$r$ is the number of processed nodes, normalized to the number of processed nodes by the sequential algorithm.
		The speedup is computed relative to the sequential algorithm.
	}
	\label{tab:knapsack-stats}
	\begin{tabular}{l | S[table-format=1.2] S[table-format=1.2] S[table-format=2.2] S[table-format=3.2]}
		PQ           & {Geom. mean $r$} & {Max $r$}      & {Harm. mean speedup} & {Max. speedup}   \\
		\hline
		MQ balanced  & \bfseries 1.08   & 2.5            & \bfseries 35.65      & \bfseries 100.84 \\
		$k$-LSM 4096 & 1.12             & \bfseries 2.47 & 0.7                  & 4.85             \\
		CA-PQ        & 1.27             & 4.59           & 2.9                  & 8.8              \\
		SMQ          & 1.11             & 4.71           & 11.92                & 43.45
	\end{tabular}
\end{table}
Figure~\ref{fig:compare-knapsack} shows the average solving times and speedups per instance for the competitors.
For better readability, we only show the balanced MultiQueue configuration and the fastest $k$-LSM variant.
The full results are given in Appendix~\ref{appendix:knapsack}.
Table~\ref{tab:knapsack-stats} shows the aggregated number of processed nodes and speedups for the best performing PQs.

The MultiQueue performs best on average for all thread counts, followed by the SMQ.
This is reflected in the individual speedups, as the MultiQueue achieves the highest speedups with \num{128} threads across almost all instances, averaging \num{35} and reaching up to \num{100} on some instances.
The CA-PQ scales well up to \num{16} threads, after which its performance plateaus.
Even the most relaxed $k$-LSM variant barely scales---presumably due to large PQ sizes---and remains slower than the sequential algorithm for all thread counts.
Similarly, the linearizable PQs do not scale and are slower than the sequential algorithm.
Unfortunately, the Spraylist fails to solve most instances within this time limit of \qty{300}{\second}.
The best-performing PQs exhibit an overhead of about $10\%$ in the number of processed nodes, demonstrating that relaxation is an effective way to speed up branch-and-bound algorithms.
They are multiple orders of magnitude faster than the linearizable competitors and scale well.

\section{Conclusion and Future Work}\label{conclusion}
The MultiQueue demonstrates that the \emph{two-choices paradigm} is a powerful concept for designing relaxed concurrent priority queues with good scalability and accuracy.
Buffering and, more importantly, stickiness help to mitigate poor cache locality and reduce cache-coherence traffic.

Compared to linearizable queues, we observed several orders of magnitude better throughput.
Thus, applications that profit from fine-grained priority-based scheduling of tasks can be effectively parallelized using the MultiQueue.
Still, benchmarking additional workloads, such as producer-consumer scenarios, and more complex applications might be interesting for better understanding the capabilities of the MultiQueue.

Although we have focused on locked $k$-ary heaps, the MultiQueue can be combined with a wide range of internal queues for various effects.
Exploiting integer keys or relaxing the internal queues to simple bucket queues can further accelerate operation, e.g., see \cite{zhangMultiBucketQueues2024}.
Using lock-free internal queues (and buffers) renders the MultiQueue deterministically lock-free in addition to being probabilistically wait-free.

Our (over)simplified analysis goes some way to explaining the good performance and quality of the MultiQueue.
Although a tighter analysis of the basic (non-sticky) MultiQueue is now available \cite{walzerSimpleExactAnalysis2024}, several questions remain open:
What quality bounds can be proven when lifting all simplifying assumptions (hopefully, the result from \cite{walzerSimpleExactAnalysis2024} still applies), and what is the effect of the stickiness parameter $s$ (one would expect a linear dependence of the rank error).

Our analysis of branch-and-bound algorithms exemplifies that, orthogonal to the inner workings of the queues, the concept of rank error and delay helps us to analyze specific applications.
It seems interesting to extend this to further applications, e.g., shortest path computations, discrete event simulations, or computing Huffman codes.

The MultiQueue design we present in this work has some limitations that could be addressed in future work.
The quality of the MultiQueue is not adaptive to the workload and its performance degrades with very few elements in the PQs.
Further, the MultiQueue is limited to a fixed number of internal PQs, thus limiting the maximum number of threads that can concurrently access the PQ.
A promising first step to address these limitations would be to adapt the stickiness period or the number of deletion candidates dynamically depending on the workload and contention level, similar to the CA-PQ.
Dynamically adjusting the number of PQs is more challenging but still an interesting direction to explore.

The conceptual simplicity and good scalability of the MultiQueue also make it an attractive candidate for a GPU implementation.
Despite current GPUs mostly use batch parallelism, it would be interesting to see whether a design based on the MultiQueue might be useful for asynchronous, massively parallel, priority-driven computing on GPUs.

\clearpage
\bibliographystyle{ACM-Reference-Format}
\bibliography{paper}


\begin{thebibliography}{55}


\ifx \showCODEN    \undefined \def \showCODEN     #1{\unskip}     \fi
\ifx \showDOI      \undefined \def \showDOI       #1{#1}\fi
\ifx \showISBNx    \undefined \def \showISBNx     #1{\unskip}     \fi
\ifx \showISBNxiii \undefined \def \showISBNxiii  #1{\unskip}     \fi
\ifx \showISSN     \undefined \def \showISSN      #1{\unskip}     \fi
\ifx \showLCCN     \undefined \def \showLCCN      #1{\unskip}     \fi
\ifx \shownote     \undefined \def \shownote      #1{#1}          \fi
\ifx \showarticletitle \undefined \def \showarticletitle #1{#1}   \fi
\ifx \showURL      \undefined \def \showURL       {\relax}        \fi
\providecommand\bibfield[2]{#2}
\providecommand\bibinfo[2]{#2}
\providecommand\natexlab[1]{#1}
\providecommand\showeprint[2][]{arXiv:#2}

\bibitem[Alistarh et~al\mbox{.}(2018)]%
        {alistarhDistributionallyLinearizableData2018}
\bibfield{author}{\bibinfo{person}{Dan Alistarh}, \bibinfo{person}{Trevor
  Brown}, \bibinfo{person}{Justin Kopinsky}, \bibinfo{person}{Jerry~Z. Li},
  {and} \bibinfo{person}{Giorgi Nadiradze}.} \bibinfo{year}{2018}\natexlab{}.
\newblock \showarticletitle{Distributionally {{Linearizable Data Structures}}}.
  In \bibinfo{booktitle}{\emph{30th {{ACM Symposium}} on {{Parallelism}} in
  {{Algorithms}} and { {Architectures}}}} \emph{(\bibinfo{series}{{{SPAA}}})}.
\newblock
\showISBNx{978-1-4503-5799-9}
\urldef\tempurl%
\url{https://doi.org/10.1145/3210377.3210411}
\showDOI{\tempurl}


\bibitem[Alistarh et~al\mbox{.}(2017)]%
        {alistarhPowerChoicePriority2017}
\bibfield{author}{\bibinfo{person}{Dan Alistarh}, \bibinfo{person}{Justin
  Kopinsky}, \bibinfo{person}{Jerry Li}, {and} \bibinfo{person}{Giorgi
  Nadiradze}.} \bibinfo{year}{2017}\natexlab{}.
\newblock \showarticletitle{The {{Power}} of {{Choice}} in {{Priority
  Scheduling}}}. In \bibinfo{booktitle}{\emph{{{ACM Symposium}} on
  {{Principles}} of {{Distributed Computing}}}}
  \emph{(\bibinfo{series}{{{PODC}}})}.
\newblock
\showISBNx{978-1-4503-4992-5}
\urldef\tempurl%
\url{https://doi.org/10.1145/3087801.3087810}
\showDOI{\tempurl}


\bibitem[Alistarh et~al\mbox{.}(2015)]%
        {alistarhSprayListScalableRelaxed2015}
\bibfield{author}{\bibinfo{person}{Dan Alistarh}, \bibinfo{person}{Justin
  Kopinsky}, \bibinfo{person}{Jerry Li}, {and} \bibinfo{person}{Nir Shavit}.}
  \bibinfo{year}{2015}\natexlab{}.
\newblock \showarticletitle{The {{SprayList}}: A Scalable Relaxed Priority
  Queue}. In \bibinfo{booktitle}{\emph{20th {{ACM SIGPLAN Symposium}} on
  {{Principles}} and {{Practice}} of {{Parallel Programming}}}}
  \emph{(\bibinfo{series}{{{PPoPP}}})}.
\newblock
\showISBNx{978-1-4503-3205-7}
\urldef\tempurl%
\url{https://doi.org/10.1145/2688500.2688523}
\showDOI{\tempurl}


\bibitem[Attiya et~al\mbox{.}(2011)]%
        {attiyaLawsOrderExpensive2011}
\bibfield{author}{\bibinfo{person}{Hagit Attiya}, \bibinfo{person}{Rachid
  Guerraoui}, \bibinfo{person}{Danny Hendler}, \bibinfo{person}{Petr
  Kuznetsov}, \bibinfo{person}{Maged~M. Michael}, {and} \bibinfo{person}{Martin
  Vechev}.} \bibinfo{year}{2011}\natexlab{}.
\newblock \showarticletitle{Laws of Order: Expensive Synchronization in
  Concurrent Algorithms Cannot Be Eliminated}.
\newblock \bibinfo{journal}{\emph{ACM SIGPLAN Notices}} \bibinfo{volume}{46},
  \bibinfo{number}{1} (\bibinfo{date}{Jan.} \bibinfo{year}{2011}).
\newblock
\showISSN{0362-1340}
\urldef\tempurl%
\url{https://doi.org/10.1145/1925844.1926442}
\showDOI{\tempurl}


\bibitem[Azar et~al\mbox{.}(1999)]%
        {azarBalancedAllocations1999}
\bibfield{author}{\bibinfo{person}{Yossi Azar}, \bibinfo{person}{Andrei~Z.
  Broder}, \bibinfo{person}{Anna~R. Karlin}, {and} \bibinfo{person}{Eli
  Upfal}.} \bibinfo{year}{1999}\natexlab{}.
\newblock \showarticletitle{Balanced {{Allocations}}}.
\newblock \bibinfo{journal}{\emph{SIAM J. Comput.}} \bibinfo{volume}{29},
  \bibinfo{number}{1} (\bibinfo{date}{Jan.} \bibinfo{year}{1999}).
\newblock
\showISSN{0097-5397}
\urldef\tempurl%
\url{https://doi.org/10.1137/S0097539795288490}
\showDOI{\tempurl}


\bibitem[Berenbrink et~al\mbox{.}(2006)]%
        {berenbrinkBalancedAllocationsHeavily2006}
\bibfield{author}{\bibinfo{person}{Petra Berenbrink}, \bibinfo{person}{Artur
  Czumaj}, \bibinfo{person}{Angelika Steger}, {and} \bibinfo{person}{Berthold
  V{\"o}~cking}.} \bibinfo{year}{2006}\natexlab{}.
\newblock \showarticletitle{Balanced {{Allocations}}: {{The Heavily Loaded
  Case}}}.
\newblock \bibinfo{journal}{\emph{SIAM J. Comput.}} \bibinfo{volume}{35},
  \bibinfo{number}{6} (\bibinfo{date}{Jan.} \bibinfo{year}{2006}).
\newblock
\showISSN{0097-5397}
\urldef\tempurl%
\url{https://doi.org/10.1137/S009753970444435X}
\showDOI{\tempurl}


\bibitem[Bingmann(2018)]%
        {TLX}
\bibfield{author}{\bibinfo{person}{Timo Bingmann}.}
  \bibinfo{year}{2018}\natexlab{}.
\newblock \bibinfo{title}{{{TLX}}: {{Collection}} of Sophisticated {{C}}++ Data
  Structures, Algorithms, and Miscellaneous Helpers}.
\newblock
\newblock
\newblock
\shownote{\url{https://panthema.net/tlx}, retrieved {Oct.} 7, 2020}.


\bibitem[Braginsky et~al\mbox{.}(2016)]%
        {braginskyCBPQHighPerformance2016}
\bibfield{author}{\bibinfo{person}{Anastasia Braginsky},
  \bibinfo{person}{Nachshon Cohen}, {and} \bibinfo{person}{Erez Petrank}.}
  \bibinfo{year}{2016}\natexlab{}.
\newblock \showarticletitle{{{CBPQ}}: {{High Performance Lock-Free Priority
  Queue}}}. In \bibinfo{booktitle}{\emph{Euro-{{Par}} 2016: {{Parallel
  Processing}}}}. \bibinfo{pages}{460--474}.
\newblock
\showISBNx{978-3-319-43659-3}
\urldef\tempurl%
\url{https://doi.org/10.1007/978-3-319-43659-3_34}
\showDOI{\tempurl}


\bibitem[Calciu et~al\mbox{.}(2014)]%
        {calciuAdaptivePriorityQueue2014}
\bibfield{author}{\bibinfo{person}{Irina Calciu}, \bibinfo{person}{Hammurabi
  Mendes}, {and} \bibinfo{person}{Maurice Herlihy}.}
  \bibinfo{year}{2014}\natexlab{}.
\newblock \showarticletitle{The {{Adaptive Priority Queue}} with
  {{Elimination}} and {{Combining} }}. In \bibinfo{booktitle}{\emph{28th
  Symposium on Distributed {{Computing}}}} \emph{(\bibinfo{series}{DC})}.
\newblock
\showISBNx{978-3-662-45174-8}
\urldef\tempurl%
\url{https://doi.org/10.1007/978-3-662-45174-8_28}
\showDOI{\tempurl}


\bibitem[Deo and Prasad(1992)]%
        {deoParallelHeapOptimal1992}
\bibfield{author}{\bibinfo{person}{Narsingh Deo} {and} \bibinfo{person}{Sushil
  Prasad}.} \bibinfo{year}{1992}\natexlab{}.
\newblock \showarticletitle{Parallel Heap: {{An}} Optimal Parallel Priority
  Queue}.
\newblock \bibinfo{journal}{\emph{The Journal of Supercomputing}}
  \bibinfo{volume}{6}, \bibinfo{number}{1} (\bibinfo{date}{March}
  \bibinfo{year}{1992}).
\newblock
\showISSN{1573-0484}
\urldef\tempurl%
\url{https://doi.org/10.1007/BF00128644}
\showDOI{\tempurl}


\bibitem[Dhulipala et~al\mbox{.}(2017)]%
        {dhulipalaJulienneFrameworkParallel2017}
\bibfield{author}{\bibinfo{person}{Laxman Dhulipala}, \bibinfo{person}{Guy
  Blelloch}, {and} \bibinfo{person}{Julian Shun}.}
  \bibinfo{year}{2017}\natexlab{}.
\newblock \showarticletitle{Julienne: {{A Framework}} for {{Parallel Graph
  Algorithms}} Using {{ Work-efficient Bucketing}}}. In
  \bibinfo{booktitle}{\emph{29th {{ACM Symposium}} on {{Parallelism}} in
  {{Algorithms}} and { {Architectures}}}} \emph{(\bibinfo{series}{{{SPAA}}})}.
\newblock
\showISBNx{978-1-4503-4593-4}
\urldef\tempurl%
\url{https://doi.org/10.1145/3087556.3087580}
\showDOI{\tempurl}


\bibitem[Dietzfelbinger and {Meyer auf der Heide}(1993)]%
        {dietzfelbingerSimpleEfficientShared1993}
\bibfield{author}{\bibinfo{person}{Martin Dietzfelbinger} {and}
  \bibinfo{person}{Friedhelm {Meyer auf der Heide}}.}
  \bibinfo{year}{1993}\natexlab{}.
\newblock \showarticletitle{Simple, Efficient Shared Memory Simulations}. In
  \bibinfo{booktitle}{\emph{5th {{ACM}} Symposium on {{Parallel Algorithms}}
  and {{ Architectures}}}} \emph{(\bibinfo{series}{{{SPAA}}})}.
\newblock
\showISBNx{978-0-89791-599-1}
\urldef\tempurl%
\url{https://doi.org/10.1145/165231.165246}
\showDOI{\tempurl}


\bibitem[Dijkstra(2022)]%
        {dijkstraNoteTwoProblems2022}
\bibfield{author}{\bibinfo{person}{E.~W. Dijkstra}.}
  \bibinfo{year}{2022}\natexlab{}.
\newblock \showarticletitle{A {{Note}} on {{Two Problems}} in {{Connexion}}
  with {{Graphs}}}.
\newblock In \bibinfo{booktitle}{\emph{Edsger {{Wybe Dijkstra}}: {{His
  Life}},{{Work}}, and {{Legacy}}} (\bibinfo{edition}{1} ed.)}.
  Vol.~\bibinfo{volume}{45}.
\newblock
\showISBNx{978-1-4503-9773-5}


\bibitem[Ellen et~al\mbox{.}(2012)]%
        {ellenInherentSequentialityConcurrent2012}
\bibfield{author}{\bibinfo{person}{Faith Ellen}, \bibinfo{person}{Danny
  Hendler}, {and} \bibinfo{person}{Nir Shavit}.}
  \bibinfo{year}{2012}\natexlab{}.
\newblock \showarticletitle{On the Inherent Sequentiality of Concurrent
  Objects}.
\newblock \bibinfo{journal}{\emph{SIAM J. Comput.}} \bibinfo{volume}{41},
  \bibinfo{number}{3} (\bibinfo{year}{2012}).
\newblock
\urldef\tempurl%
\url{https://doi.org/10.1137/08072646X}
\showDOI{\tempurl}
\showeprint{https://doi.org/10.1137/08072646X}


\bibitem[Ellen et~al\mbox{.}(2007)]%
        {ellenSNZIScalableNonZero2007}
\bibfield{author}{\bibinfo{person}{Faith Ellen}, \bibinfo{person}{Yossi Lev},
  \bibinfo{person}{Victor Luchangco}, {and} \bibinfo{person}{Mark Moir}.}
  \bibinfo{year}{2007}\natexlab{}.
\newblock \showarticletitle{{{SNZI}}: Scalable {{NonZero}} Indicators}. In
  \bibinfo{booktitle}{\emph{{{ACM}} Symposium on {{Principles}} of Distributed
  Computing}}. \bibinfo{pages}{13--22}.
\newblock
\showISBNx{978-1-59593-616-5}
\urldef\tempurl%
\url{https://doi.org/10.1145/1281100.1281106}
\showDOI{\tempurl}


\bibitem[Funke et~al\mbox{.}(2019)]%
        {funkeCommunicationfreeMassivelyDistributed2019}
\bibfield{author}{\bibinfo{person}{Daniel Funke}, \bibinfo{person}{Sebastian
  Lamm}, \bibinfo{person}{Ulrich Meyer}, \bibinfo{person}{Manuel Penschuck},
  \bibinfo{person}{Peter Sanders}, \bibinfo{person}{Christian Schulz},
  \bibinfo{person}{Darren Strash}, {and} \bibinfo{person}{Moritz {von Looz}}.}
  \bibinfo{year}{2019}\natexlab{}.
\newblock \showarticletitle{Communication-Free Massively Distributed Graph
  Generation}.
\newblock \bibinfo{journal}{\emph{J. Parallel and Distrib. Comput.}}
  \bibinfo{volume}{131} (\bibinfo{date}{Sept.} \bibinfo{year}{2019}).
\newblock
\showISSN{0743-7315}
\urldef\tempurl%
\url{https://doi.org/10.1016/j.jpdc.2019.03.011}
\showDOI{\tempurl}


\bibitem[Haas et~al\mbox{.}(2013)]%
        {haasDistributedQueuesShared2013}
\bibfield{author}{\bibinfo{person}{Andreas Haas}, \bibinfo{person}{Michael
  Lippautz}, \bibinfo{person}{Thomas~A. Henzinger}, \bibinfo{person}{Hannes
  Payer}, \bibinfo{person}{Ana Sokolova}, \bibinfo{person}{Christoph~M.
  Kirsch}, {and} \bibinfo{person}{Ali Sezgin}.}
  \bibinfo{year}{2013}\natexlab{}.
\newblock \showarticletitle{Distributed Queues in Shared Memory: Multicore
  Performance and Scalability through Quantitative Relaxation}. In
  \bibinfo{booktitle}{\emph{{{ACM International Conference}} on {{Computing
  Frontiers}}}} \emph{(\bibinfo{series}{{{CF}}})}.
\newblock
\showISBNx{978-1-4503-2053-5}
\urldef\tempurl%
\url{https://doi.org/10.1145/2482767.2482789}
\showDOI{\tempurl}


\bibitem[Henzinger et~al\mbox{.}(2013)]%
        {henzingerQuantitativeRelaxationConcurrent2013}
\bibfield{author}{\bibinfo{person}{Thomas~A. Henzinger},
  \bibinfo{person}{Christoph~M. Kirsch}, \bibinfo{person}{Hannes Payer},
  \bibinfo{person}{Ali Sezgin}, {and} \bibinfo{person}{Ana Sokolova}.}
  \bibinfo{year}{2013}\natexlab{}.
\newblock \showarticletitle{Quantitative Relaxation of Concurrent Data
  Structures}.
\newblock \bibinfo{journal}{\emph{ACM SIGPLAN Notices}} \bibinfo{volume}{48},
  \bibinfo{number}{1} (\bibinfo{date}{Jan.} \bibinfo{year}{2013}).
\newblock
\showISSN{0362-1340}
\urldef\tempurl%
\url{https://doi.org/10.1145/2480359.2429109}
\showDOI{\tempurl}


\bibitem[Herlihy et~al\mbox{.}(2021)]%
        {herlihyArtMultiprocessorProgramming2020}
\bibfield{author}{\bibinfo{person}{Maurice Herlihy}, \bibinfo{person}{Nir
  Shavit}, \bibinfo{person}{Victor Luchangco}, {and} \bibinfo{person}{Michael
  Spear}.} \bibinfo{year}{2021}\natexlab{}.
\newblock \bibinfo{booktitle}{\emph{The Art of Multiprocessor Programming}}.
\newblock \bibinfo{publisher}{Morgan Kaufmann}.
\newblock
\showISBNx{978-0-12-415950-1}
\urldef\tempurl%
\url{https://doi.org/10.1016/C2011-0-06993-4}
\showDOI{\tempurl}


\bibitem[Herlihy and Wing(1990)]%
        {herlihyLinearizabilityCorrectnessCondition1990}
\bibfield{author}{\bibinfo{person}{Maurice~P. Herlihy} {and}
  \bibinfo{person}{Jeannette~M. Wing}.} \bibinfo{year}{1990}\natexlab{}.
\newblock \showarticletitle{Linearizability: A Correctness Condition for
  Concurrent Objects}.
\newblock \bibinfo{journal}{\emph{ACM Transactions on Programming Languages and
  Systems}} \bibinfo{volume}{12}, \bibinfo{number}{3} (\bibinfo{date}{July}
  \bibinfo{year}{1990}).
\newblock
\showISSN{0164-0925}
\urldef\tempurl%
\url{https://doi.org/10.1145/78969.78972}
\showDOI{\tempurl}


\bibitem[{H{\"u}bschle-Schneider} and Sanders(2016)]%
        {hubschle-schneiderCommunicationEfficientAlgorithms2016}
\bibfield{author}{\bibinfo{person}{Lorenz {H{\"u}bschle-Schneider}} {and}
  \bibinfo{person}{Peter Sanders}.} \bibinfo{year}{2016}\natexlab{}.
\newblock \showarticletitle{Communication {{Efficient Algorithms}} for {{Top-k
  Selection Problems }}}. In \bibinfo{booktitle}{\emph{{{IEEE International
  Parallel}} and {{Distributed Processing Symposium}}}}
  \emph{(\bibinfo{series}{{{IPDPS}}})}.
\newblock
\showISSN{1530-2075}
\urldef\tempurl%
\url{https://doi.org/10.1109/IPDPS.2016.45}
\showDOI{\tempurl}


\bibitem[Johnson(1975)]%
        {johnsonPriorityQueuesUpdate1975}
\bibfield{author}{\bibinfo{person}{Donald~B. Johnson}.}
  \bibinfo{year}{1975}\natexlab{}.
\newblock \showarticletitle{Priority Queues with Update and Finding Minimum
  Spanning Trees}.
\newblock \bibinfo{journal}{\emph{Inform. Process. Lett.}} \bibinfo{volume}{4},
  \bibinfo{number}{3} (\bibinfo{date}{Dec.} \bibinfo{year}{1975}).
\newblock
\showISSN{0020-0190}
\urldef\tempurl%
\url{https://doi.org/10.1016/0020-0190(75)90001-0}
\showDOI{\tempurl}


\bibitem[Karp and Zhang(1993)]%
        {karpRandomizedParallelAlgorithms1993}
\bibfield{author}{\bibinfo{person}{Richard~M. Karp} {and}
  \bibinfo{person}{Yanjun Zhang}.} \bibinfo{year}{1993}\natexlab{}.
\newblock \showarticletitle{Randomized Parallel Algorithms for Backtrack Search
  and Branch-and-Bound Computation}.
\newblock \bibinfo{journal}{\emph{J. ACM}} \bibinfo{volume}{40},
  \bibinfo{number}{3} (\bibinfo{date}{July} \bibinfo{year}{1993}).
\newblock
\showISSN{0004-5411, 1557-735X}
\urldef\tempurl%
\url{https://doi.org/10.1145/174130.174145}
\showDOI{\tempurl}


\bibitem[Lind{\'e}n and Jonsson(2013)]%
        {lindenSkiplistBasedConcurrentPriority2013}
\bibfield{author}{\bibinfo{person}{Jonatan Lind{\'e}n} {and}
  \bibinfo{person}{Bengt Jonsson}.} \bibinfo{year}{2013}\natexlab{}.
\newblock \showarticletitle{A {{Skiplist-Based Concurrent Priority Queue}} with
  {{Minimal Memory Contention}}}. In \bibinfo{booktitle}{\emph{Principles of
  {{Distributed Systems}}}} \emph{(\bibinfo{series}{{{OPODIS}}})}.
\newblock
\showISBNx{978-3-319-03850-6}
\urldef\tempurl%
\url{https://doi.org/10.1007/978-3-319-03850-6_15}
\showDOI{\tempurl}


\bibitem[Maier et~al\mbox{.}(2018)]%
        {maierConcurrentHashTables2018}
\bibfield{author}{\bibinfo{person}{Tobias Maier}, \bibinfo{person}{Peter
  Sanders}, {and} \bibinfo{person}{Roman Dementiev}.}
  \bibinfo{year}{2018}\natexlab{}.
\newblock \showarticletitle{Concurrent {{Hash Tables}}: {{Fast}} and
  {{General}}(?)!}
\newblock \bibinfo{journal}{\emph{ACM Transactions on Parallel Computing}}
  \bibinfo{volume}{5}, \bibinfo{number}{4} (\bibinfo{date}{Dec.}
  \bibinfo{year}{2018}).
\newblock
\showISSN{2329-4949, 2329-4957}
\urldef\tempurl%
\url{https://doi.org/10.1145/3309206}
\showDOI{\tempurl}


\bibitem[Martello and Toth(1990)]%
        {martelloKnapsackProblemsAlgorithms1990}
\bibfield{author}{\bibinfo{person}{Silvano Martello} {and}
  \bibinfo{person}{Paolo Toth}.} \bibinfo{year}{1990}\natexlab{}.
\newblock \bibinfo{booktitle}{\emph{Knapsack Problems: Algorithms and Computer
  Implementations}}.
\newblock \bibinfo{publisher}{John Wiley \& Sons, Inc.}
\newblock
\showISBNx{978-0-471-92420-3}


\bibitem[Matocha and Camp(1998)]%
        {matochaTaxonomyDistributedTermination1998}
\bibfield{author}{\bibinfo{person}{Jeff Matocha} {and} \bibinfo{person}{Tracy
  Camp}.} \bibinfo{year}{1998}\natexlab{}.
\newblock \showarticletitle{A Taxonomy of Distributed Termination Detection
  Algorithms}.
\newblock \bibinfo{journal}{\emph{Journal of Systems and Software}}
  \bibinfo{volume}{43}, \bibinfo{number}{3} (\bibinfo{date}{Nov.}
  \bibinfo{year}{1998}).
\newblock
\showISSN{0164-1212}
\urldef\tempurl%
\url{https://doi.org/10.1016/S0164-1212(98)10034-1}
\showDOI{\tempurl}


\bibitem[Mattern(1987)]%
        {matternAlgorithmsDistributedTermination1987}
\bibfield{author}{\bibinfo{person}{Friedemann Mattern}.}
  \bibinfo{year}{1987}\natexlab{}.
\newblock \showarticletitle{Algorithms for Distributed Termination Detection}.
\newblock \bibinfo{journal}{\emph{Distributed Computing}} \bibinfo{volume}{2},
  \bibinfo{number}{3} (\bibinfo{date}{Sept.} \bibinfo{year}{1987}).
\newblock
\showISSN{1432-0452}
\urldef\tempurl%
\url{https://doi.org/10.1007/BF01782776}
\showDOI{\tempurl}


\bibitem[Mehlhorn and N{\"a}her(1990)]%
        {mehlhornBoundedOrderedDictionaries1990}
\bibfield{author}{\bibinfo{person}{Kurt Mehlhorn} {and} \bibinfo{person}{Stefan
  N{\"a}her}.} \bibinfo{year}{1990}\natexlab{}.
\newblock \showarticletitle{Bounded Ordered Dictionaries in {{O}}(log log
  {{{\emph{N}}}}) Time and {{O}}({\emph{n}}) Space}.
\newblock \bibinfo{journal}{\emph{Inform. Process. Lett.}}
  \bibinfo{volume}{35}, \bibinfo{number}{4} (\bibinfo{date}{Aug.}
  \bibinfo{year}{1990}).
\newblock
\showISSN{0020-0190}
\urldef\tempurl%
\url{https://doi.org/10.1016/0020-0190(90)90022-P}
\showDOI{\tempurl}


\bibitem[Meyer and Sanders(2003)]%
        {meyerDsteppingParallelizableShortest2003}
\bibfield{author}{\bibinfo{person}{U. Meyer} {and} \bibinfo{person}{P.
  Sanders}.} \bibinfo{year}{2003}\natexlab{}.
\newblock \showarticletitle{{$\Delta$}-Stepping: A Parallelizable Shortest Path
  Algorithm}.
\newblock \bibinfo{journal}{\emph{Journal of Algorithms}} \bibinfo{volume}{49},
  \bibinfo{number}{1} (\bibinfo{date}{Oct.} \bibinfo{year}{2003}).
\newblock
\showISSN{0196-6774}
\urldef\tempurl%
\url{https://doi.org/10.1016/S0196-6774(03)00076-2}
\showDOI{\tempurl}


\bibitem[Mitzenmacher et~al\mbox{.}(2001)]%
        {mitzenmacherPowerTwoRandom2001}
\bibfield{author}{\bibinfo{person}{Michael Mitzenmacher},
  \bibinfo{person}{Andr{\'e}a~W. Richa}, {and} \bibinfo{person}{Ramesh
  Sitaraman}.} \bibinfo{year}{2001}\natexlab{}.
\newblock \showarticletitle{The {{Power}} of {{Two Random Choices}}: {{A
  Survey}} of {{Techniques }} and {{Results}}}.
\newblock In \bibinfo{booktitle}{\emph{Handbook of {{Randomized Computing}}}}.
  Vol.~\bibinfo{volume}{9}.
\newblock
\showISBNx{978-1-4613-4886-3 978-1-4615-0013-1}
\urldef\tempurl%
\url{https://doi.org/10.1007/978-1-4615-0013-1_9}
\showDOI{\tempurl}


\bibitem[Nguyen et~al\mbox{.}(2013)]%
        {nguyenLightweightInfrastructureGraph2013}
\bibfield{author}{\bibinfo{person}{Donald Nguyen}, \bibinfo{person}{Andrew
  Lenharth}, {and} \bibinfo{person}{Keshav Pingali}.}
  \bibinfo{year}{2013}\natexlab{}.
\newblock \showarticletitle{A Lightweight Infrastructure for Graph Analytics}.
  In \bibinfo{booktitle}{\emph{24th {{ACM Symposium}} on {{Operating Systems
  Principles}}}} \emph{(\bibinfo{series}{{{SOSP}}})}.
\newblock
\showISBNx{978-1-4503-2388-8}
\urldef\tempurl%
\url{https://doi.org/10.1145/2517349.2522739}
\showDOI{\tempurl}


\bibitem[O'Neil et~al\mbox{.}(1996)]%
        {oneilLogstructuredMergetreeLSMtree1996}
\bibfield{author}{\bibinfo{person}{Patrick O'Neil}, \bibinfo{person}{Edward
  Cheng}, \bibinfo{person}{Dieter Gawlick}, {and} \bibinfo{person}{Elizabeth
  O'Neil}.} \bibinfo{year}{1996}\natexlab{}.
\newblock \showarticletitle{The Log-Structured Merge-Tree ({{LSM-tree}})}.
\newblock \bibinfo{journal}{\emph{Acta Informatica}} \bibinfo{volume}{33},
  \bibinfo{number}{4} (\bibinfo{date}{June} \bibinfo{year}{1996}).
\newblock
\showISSN{1432-0525}
\urldef\tempurl%
\url{https://doi.org/10.1007/s002360050048}
\showDOI{\tempurl}


\bibitem[Postnikova et~al\mbox{.}(2022)]%
        {postnikovaMultiqueuesCanBe2022}
\bibfield{author}{\bibinfo{person}{Anastasiia Postnikova},
  \bibinfo{person}{Nikita Koval}, \bibinfo{person}{Giorgi Nadiradze}, {and}
  \bibinfo{person}{Dan Alistarh}.} \bibinfo{year}{2022}\natexlab{}.
\newblock \showarticletitle{Multi-Queues Can Be State-of-the-Art Priority
  Schedulers}. In \bibinfo{booktitle}{\emph{27th {{ACM SIGPLAN Symposium}} on
  {{Principles}} and {{Practice}} of {{Parallel Programming}}}}
  \emph{(\bibinfo{series}{{{PPoPP}}})}.
\newblock
\showISBNx{978-1-4503-9204-4}
\urldef\tempurl%
\url{https://doi.org/10.1145/3503221.3508432}
\showDOI{\tempurl}


\bibitem[Pugh(1990)]%
        {pughSkipListsProbabilistic1990}
\bibfield{author}{\bibinfo{person}{William Pugh}.}
  \bibinfo{year}{1990}\natexlab{}.
\newblock \showarticletitle{Skip Lists: A Probabilistic Alternative to Balanced
  Trees}.
\newblock \bibinfo{journal}{\emph{Commun. ACM}} \bibinfo{volume}{33},
  \bibinfo{number}{6} (\bibinfo{date}{June} \bibinfo{year}{1990}).
\newblock
\showISSN{0001-0782}
\urldef\tempurl%
\url{https://doi.org/10.1145/78973.78977}
\showDOI{\tempurl}


\bibitem[Raab and Steger(1998)]%
        {raabBallsBinsSimple1998}
\bibfield{author}{\bibinfo{person}{Martin Raab} {and} \bibinfo{person}{Angelika
  Steger}.} \bibinfo{year}{1998}\natexlab{}.
\newblock \showarticletitle{``{{Balls}} into {{Bins}}'' --- {{A Simple}} and
  {{Tight Analysis}}}. In \bibinfo{booktitle}{\emph{Randomization and
  {{Approximation Techniques}} in {{Computer Science}}}}.
\newblock
\showISBNx{978-3-540-49543-7}
\urldef\tempurl%
\url{https://doi.org/10.1007/3-540-49543-6_13}
\showDOI{\tempurl}


\bibitem[Rihani et~al\mbox{.}(2014)]%
        {rihaniMultiQueuesSimplerFaster2014}
\bibfield{author}{\bibinfo{person}{Hamza Rihani}, \bibinfo{person}{Peter
  Sanders}, {and} \bibinfo{person}{Roman Dementiev}.}
  \bibinfo{year}{2014}\natexlab{}.
\newblock \bibinfo{title}{{{MultiQueues}}: {{Simpler}}, {{Faster}}, and
  {{Better Relaxed Concurrent Priority Queues}}}.
\newblock
\newblock
\urldef\tempurl%
\url{https://doi.org/10.48550/arXiv.1411.1209}
\showDOI{\tempurl}


\bibitem[Rihani et~al\mbox{.}(2015)]%
        {rihaniMultiQueuesSimpleRelaxed2015}
\bibfield{author}{\bibinfo{person}{Hamza Rihani}, \bibinfo{person}{Peter
  Sanders}, {and} \bibinfo{person}{Roman Dementiev}.}
  \bibinfo{year}{2015}\natexlab{}.
\newblock \showarticletitle{{{MultiQueues}}: {{Simple Relaxed Concurrent
  Priority Queues}}}. In \bibinfo{booktitle}{\emph{27th {{ACM}} Symposium on
  {{Parallelism}} in {{Algorithms}} and { {Architectures}}}}
  \emph{(\bibinfo{series}{{{SPAA}}})}.
\newblock
\showISBNx{978-1-4503-3588-1}
\urldef\tempurl%
\url{https://doi.org/10.1145/2755573.2755616}
\showDOI{\tempurl}


\bibitem[Rukundo et~al\mbox{.}(2019)]%
        {rukundoMonotonicallyRelaxingConcurrent2019}
\bibfield{author}{\bibinfo{person}{Adones Rukundo}, \bibinfo{person}{Aras
  Atalar}, {and} \bibinfo{person}{Philippas Tsigas}.}
  \bibinfo{year}{2019}\natexlab{}.
\newblock \showarticletitle{Monotonically {{Relaxing Concurrent Data-Structure
  Semantics}} for {{ Increasing Performance}}: {{An Efficient 2D Design
  Framework}}}. In \bibinfo{booktitle}{\emph{33rd {{Symposium}} on
  {{Distributed Computing}}}} \emph{(\bibinfo{series}{{{DISC}}})}.
\newblock
\showISBNx{978-3-95977-126-9}
\showISSN{1868-8969}
\urldef\tempurl%
\url{https://doi.org/10.4230/LIPIcs.DISC.2019.31}
\showDOI{\tempurl}


\bibitem[Rukundo and Tsigas(2021)]%
        {rukundoTSLQueueEfficientLockFree2021}
\bibfield{author}{\bibinfo{person}{Adones Rukundo} {and}
  \bibinfo{person}{Philippas Tsigas}.} \bibinfo{year}{2021}\natexlab{}.
\newblock \showarticletitle{{{TSLQueue}}: {{An Efficient Lock-Free Design}} for
  {{Priority Queues }}}. In \bibinfo{booktitle}{\emph{Euro-{{Par}}: {{Parallel
  Processing}}}}.
\newblock
\showISBNx{978-3-030-85665-6}
\urldef\tempurl%
\url{https://doi.org/10.1007/978-3-030-85665-6_24}
\showDOI{\tempurl}


\bibitem[Sagonas and Winblad(2017)]%
        {sagonasContentionAvoidingConcurrent2017}
\bibfield{author}{\bibinfo{person}{Konstantinos Sagonas} {and}
  \bibinfo{person}{Kjell Winblad}.} \bibinfo{year}{2017}\natexlab{}.
\newblock \showarticletitle{The {{Contention Avoiding Concurrent Priority
  Queue}}}. In \bibinfo{booktitle}{\emph{Languages and {{Compilers}} for
  {{Parallel Computing}}}}.
\newblock
\showISBNx{978-3-319-52709-3}
\urldef\tempurl%
\url{https://doi.org/10.1007/978-3-319-52709-3_23}
\showDOI{\tempurl}


\bibitem[Sanders(1995)]%
        {sandersFastPriorityQueues1995}
\bibfield{author}{\bibinfo{person}{Peter Sanders}.}
  \bibinfo{year}{1995}\natexlab{}.
\newblock \showarticletitle{Fast Priority Queues for Parallel
  Branch-and-Bound}. In \bibinfo{booktitle}{\emph{Parallel {{Algorithms}} for
  {{Irregularly Structured Problems}}}}.
\newblock
\showISBNx{978-3-540-44915-7}
\urldef\tempurl%
\url{https://doi.org/10.1007/3-540-60321-2_30}
\showDOI{\tempurl}


\bibitem[Sanders(1997)]%
        {sandersLastverteilungsalgorithmenFuerParallele1997}
\bibfield{author}{\bibinfo{person}{Peter Sanders}.}
  \bibinfo{year}{1997}\natexlab{}.
\newblock \emph{\bibinfo{title}{{Lastverteilungsalgorithmen f{\"u}r parallele
  Tiefensuche}}}.
\newblock \bibinfo{thesistype}{Ph.\,D. Dissertation}.
\newblock
\showISBNx{9783183463107}
\showISSN{0178-9546}
\urldef\tempurl%
\url{https://doi.org/10.5445/IR/997}
\showDOI{\tempurl}


\bibitem[Sanders(1998)]%
        {sandersRandomizedPriorityQueues1998}
\bibfield{author}{\bibinfo{person}{Peter Sanders}.}
  \bibinfo{year}{1998}\natexlab{}.
\newblock \showarticletitle{Randomized {{Priority Queues}} for {{Fast Parallel
  Access}}}.
\newblock \bibinfo{journal}{\emph{J. Parallel and Distrib. Comput.}}
  \bibinfo{volume}{49}, \bibinfo{number}{1} (\bibinfo{date}{Feb.}
  \bibinfo{year}{1998}).
\newblock
\showISSN{0743-7315}
\urldef\tempurl%
\url{https://doi.org/10.1006/jpdc.1998.1429}
\showDOI{\tempurl}


\bibitem[Sanders(2001)]%
        {sandersFastPriorityQueues2001}
\bibfield{author}{\bibinfo{person}{Peter Sanders}.}
  \bibinfo{year}{2001}\natexlab{}.
\newblock \showarticletitle{Fast Priority Queues for Cached Memory}.
\newblock \bibinfo{journal}{\emph{ACM Journal of Experimental Algorithmics}}
  \bibinfo{volume}{5} (\bibinfo{date}{Dec.} \bibinfo{year}{2001}).
\newblock
\showISSN{1084-6654}
\urldef\tempurl%
\url{https://doi.org/10.1145/351827.384249}
\showDOI{\tempurl}


\bibitem[Sanders et~al\mbox{.}(2019)]%
        {sandersSequentialParallelAlgorithms2019}
\bibfield{author}{\bibinfo{person}{Peter Sanders}, \bibinfo{person}{Kurt
  Mehlhorn}, \bibinfo{person}{Martin Dietzfelbinger}, {and}
  \bibinfo{person}{Roman Dementiev}.} \bibinfo{year}{2019}\natexlab{}.
\newblock \bibinfo{booktitle}{\emph{Sequential and {{Parallel Algorithms}} and
  {{Data Structures}}: {{The Basic Toolbox}}}}.
\newblock \bibinfo{publisher}{Springer}.
\newblock
\showISBNx{978-3-030-25208-3 978-3-030-25209-0}
\urldef\tempurl%
\url{https://doi.org/10.1007/978-3-030-25209-0}
\showDOI{\tempurl}


\bibitem[Shavit and Lotan(2000)]%
        {shavitSkiplistbasedConcurrentPriority2000}
\bibfield{author}{\bibinfo{person}{N. Shavit} {and} \bibinfo{person}{I.
  Lotan}.} \bibinfo{year}{2000}\natexlab{}.
\newblock \showarticletitle{Skiplist-Based Concurrent Priority Queues}. In
  \bibinfo{booktitle}{\emph{14th {{International Parallel}} and {{Distributed
  Processing Symposium}}.}} \emph{(\bibinfo{series}{{{IPDPS}}})}.
\newblock
\urldef\tempurl%
\url{https://doi.org/10.1109/IPDPS.2000.845994}
\showDOI{\tempurl}


\bibitem[Shier and Witzgall(1981)]%
        {shierPropertiesLabelingMethods1981}
\bibfield{author}{\bibinfo{person}{D.R. Shier} {and} \bibinfo{person}{C.
  Witzgall}.} \bibinfo{year}{1981}\natexlab{}.
\newblock \showarticletitle{Properties of {{Labeling Methods}} for
  {{Determining Shortest Path Trees}}}.
\newblock \bibinfo{journal}{\emph{J. Res. Nat. Bur. Standards}}
  \bibinfo{volume}{86}, \bibinfo{number}{3} (\bibinfo{year}{1981}),
  \bibinfo{pages}{317}.
\newblock
\showISSN{0160-1741}
\urldef\tempurl%
\url{https://doi.org/10.6028/jres.086.013}
\showDOI{\tempurl}


\bibitem[{van Emde Boas} et~al\mbox{.}(1976)]%
        {vanemdeboasDesignImplementationEfficient1976}
\bibfield{author}{\bibinfo{person}{P. {van Emde Boas}}, \bibinfo{person}{R.
  Kaas}, {and} \bibinfo{person}{E. Zijlstra}.} \bibinfo{year}{1976}\natexlab{}.
\newblock \showarticletitle{Design and Implementation of an Efficient Priority
  Queue}.
\newblock \bibinfo{journal}{\emph{Mathematical systems theory}}
  \bibinfo{volume}{10}, \bibinfo{number}{1} (\bibinfo{date}{Dec.}
  \bibinfo{year}{1976}).
\newblock
\showISSN{1433-0490}
\urldef\tempurl%
\url{https://doi.org/10.1007/BF01683268}
\showDOI{\tempurl}


\bibitem[Walzer and Williams(2024)]%
        {walzerSimpleExactAnalysis2024}
\bibfield{author}{\bibinfo{person}{Stefan Walzer} {and} \bibinfo{person}{Marvin
  Williams}.} \bibinfo{year}{2024}\natexlab{}.
\newblock \bibinfo{title}{A {{Simple}} yet {{Exact Analysis}} of the
  {{MultiQueue}}}.
\newblock
\newblock
\urldef\tempurl%
\url{https://doi.org/10.48550/arXiv.2410.08714}
\showDOI{\tempurl}
\showeprint[arxiv]{2410.08714}~[cs]


\bibitem[Williams(1964)]%
        {williamsHeapsort1964}
\bibfield{author}{\bibinfo{person}{J.~W.~J. Williams}.}
  \bibinfo{year}{1964}\natexlab{}.
\newblock \showarticletitle{Algorithm 232: Heapsort}.
\newblock \bibinfo{journal}{\emph{Commun. ACM}} \bibinfo{volume}{7},
  \bibinfo{number}{6} (\bibinfo{year}{1964}).
\newblock
\urldef\tempurl%
\url{https://doi.org/10.1145/512274.512284}
\showDOI{\tempurl}


\bibitem[Williams et~al\mbox{.}(2021)]%
        {williamsEngineeringMultiQueuesFast2021}
\bibfield{author}{\bibinfo{person}{Marvin Williams}, \bibinfo{person}{Peter
  Sanders}, {and} \bibinfo{person}{Roman Dementiev}.}
  \bibinfo{year}{2021}\natexlab{}.
\newblock \showarticletitle{Engineering {{MultiQueues}}: {{Fast Relaxed
  Concurrent Priority Queues}}}. In \bibinfo{booktitle}{\emph{29th {{European
  Symposium}} on {{Algorithms}}}} \emph{(\bibinfo{series}{{{ESA}}})}.
\newblock
\showISBNx{978-3-95977-204-4}
\urldef\tempurl%
\url{https://doi.org/10.4230/LIPIcs.ESA.2021.81}
\showDOI{\tempurl}


\bibitem[Wimmer et~al\mbox{.}(2015)]%
        {wimmerLockfreeKLSMRelaxed2015}
\bibfield{author}{\bibinfo{person}{Martin Wimmer}, \bibinfo{person}{Jakob
  Gruber}, \bibinfo{person}{Jesper~Larsson Tr{\"a}ff}, {and}
  \bibinfo{person}{Philippas Tsigas}.} \bibinfo{year}{2015}\natexlab{}.
\newblock \showarticletitle{The Lock-Free k-{{LSM}} Relaxed Priority Queue}.
\newblock \bibinfo{journal}{\emph{ACM SIGPLAN Notices}} \bibinfo{volume}{50},
  \bibinfo{number}{8} (\bibinfo{date}{Jan.} \bibinfo{year}{2015}).
\newblock
\showISSN{0362-1340}
\urldef\tempurl%
\url{https://doi.org/10.1145/2858788.2688547}
\showDOI{\tempurl}


\bibitem[Yesil et~al\mbox{.}(2019)]%
        {yesilUnderstandingPrioritybasedScheduling2019}
\bibfield{author}{\bibinfo{person}{Serif Yesil}, \bibinfo{person}{Azin
  Heidarshenas}, \bibinfo{person}{Adam Morrison}, {and} \bibinfo{person}{Josep
  Torrellas}.} \bibinfo{year}{2019}\natexlab{}.
\newblock \showarticletitle{Understanding Priority-Based Scheduling of Graph
  Algorithms on a Shared-Memory Platform}. In
  \bibinfo{booktitle}{\emph{Conference for {{High Performance Computing}},
  {{Networking}}, {{ Storage}} and {{Analysis}}}}
  \emph{(\bibinfo{series}{{{SC}}})}.
\newblock
\showISBNx{978-1-4503-6229-0}
\urldef\tempurl%
\url{https://doi.org/10.1145/3295500.3356160}
\showDOI{\tempurl}


\bibitem[Zhang et~al\mbox{.}(2024)]%
        {zhangMultiBucketQueues2024}
\bibfield{author}{\bibinfo{person}{Guozheng Zhang}, \bibinfo{person}{Gilead
  Posluns}, {and} \bibinfo{person}{Mark~C. Jeffrey}.}
  \bibinfo{year}{2024}\natexlab{}.
\newblock \showarticletitle{Multi {{Bucket Queues}}: {{Efficient Concurrent
  Priority Scheduling}}}. In \bibinfo{booktitle}{\emph{36th {{ACM Symposium}}
  on {{Parallelism}} in {{Algorithms}} and { {Architectures}}}}
  \emph{(\bibinfo{series}{{{SPAA}}})}.
\newblock
\showISBNx{9798400704161}
\urldef\tempurl%
\url{https://doi.org/10.1145/3626183.3659962}
\showDOI{\tempurl}


\end{thebibliography}
\clearpage
\appendix
\section{Throughput vs. Quality}\label{appendix:throughput_quality}
\begin{figure}[ht]
	\includegraphics{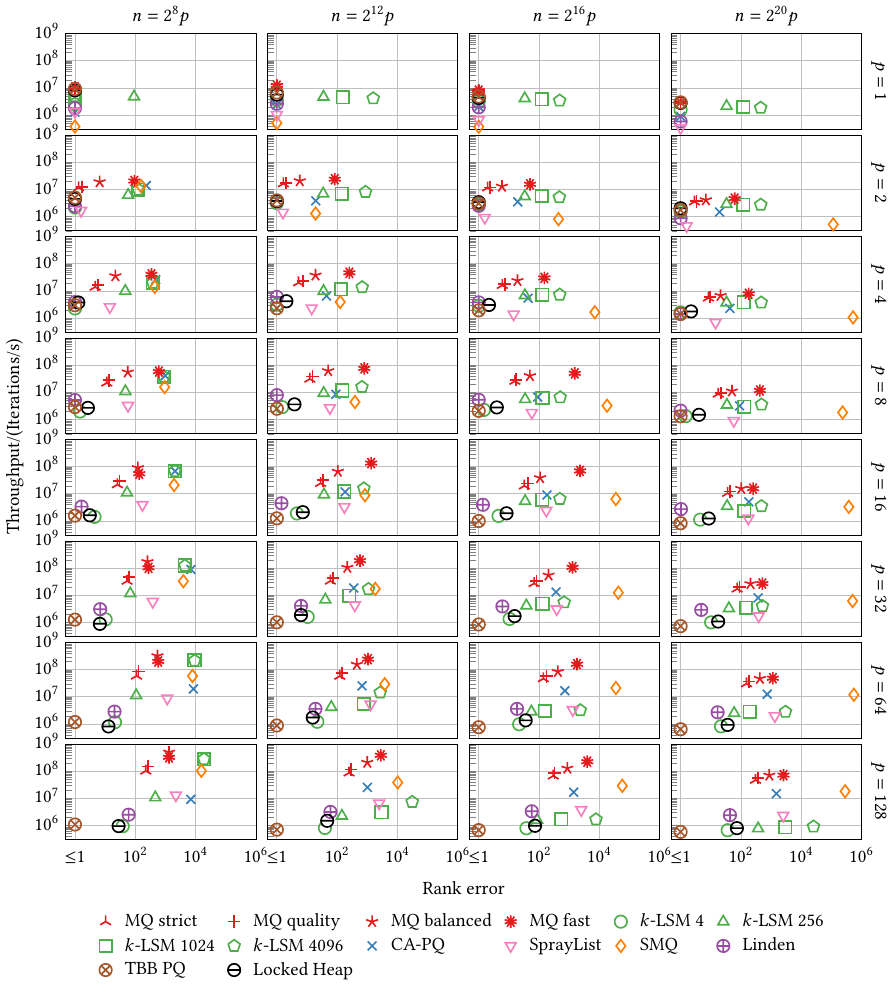}
	\caption{
		Mean rank error vs. throughput for the monotonic stress test with different pre-fills $n$ on machine \textsf{AMD}.
		Each row corresponds to a different thread count ($p$) from $1,2,4,\ldots,128$ and each column to a different pre-fill.
		Tick labels are omitted due to space limitations.
	}
	\label{fig:comparison-throughput-quality-full}
\end{figure}
\section{Throughput Scaling}\label{appendix:scaling}
\begin{figure}[ht]
	\includegraphics{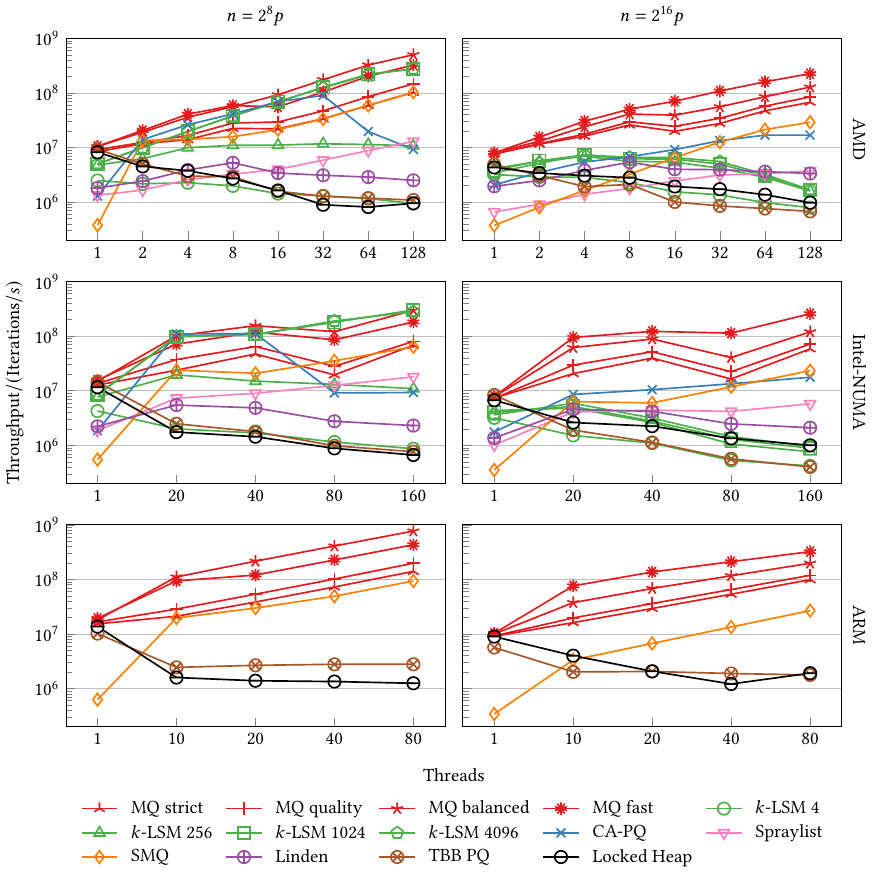}
	\caption{
		Weak scaling experiment showing throughput for the monotonic stress test versus the number of threads on different machines.
		Each row corresponds to a different machine and each column to a different pre-fill.
	}
\end{figure}
\clearpage
\section{Insert-Delete Stress Test}\label{appendix:insdel}
\begin{figure}[ht]
	\includegraphics{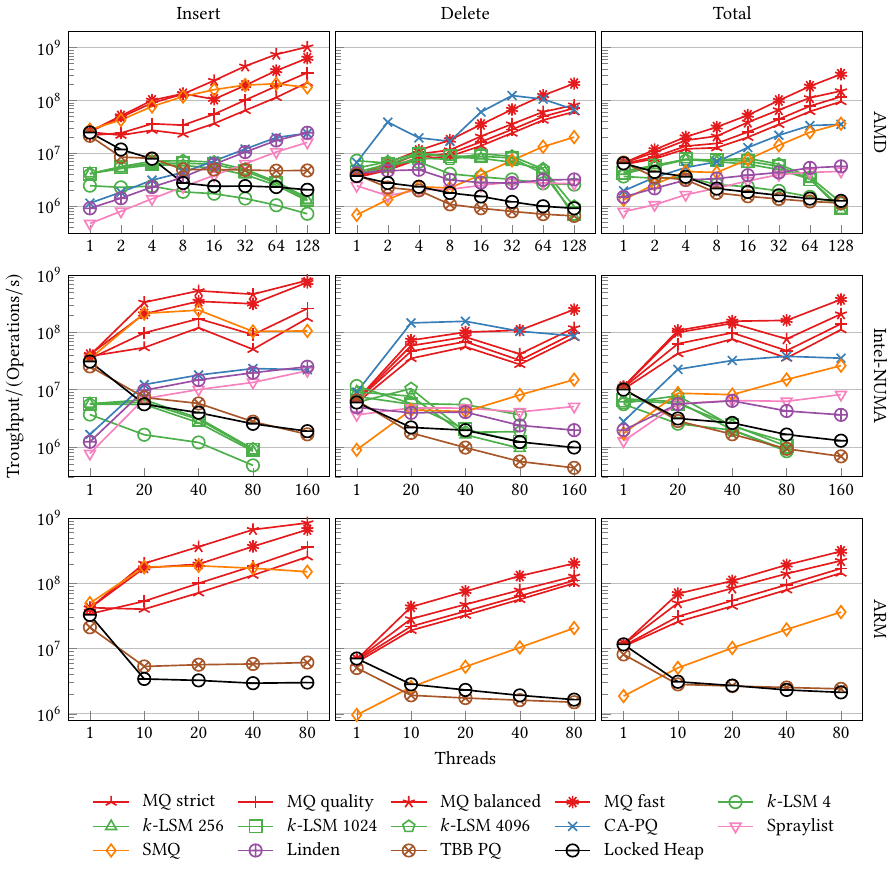}
	\caption{
		Throughput comparison of the insert-delete stress test with $2^{20}$ elements per thread on different machines.
		Data points for the $k$-LSM variants on machine \textsf{Intel-NUMA} for \num{160} threads are missing due to crashes.
	}
\end{figure}
\clearpage
\section{Single-Source Shortest Path}\label{appendix:sssp}
\begin{figure}[ht]
	\includegraphics{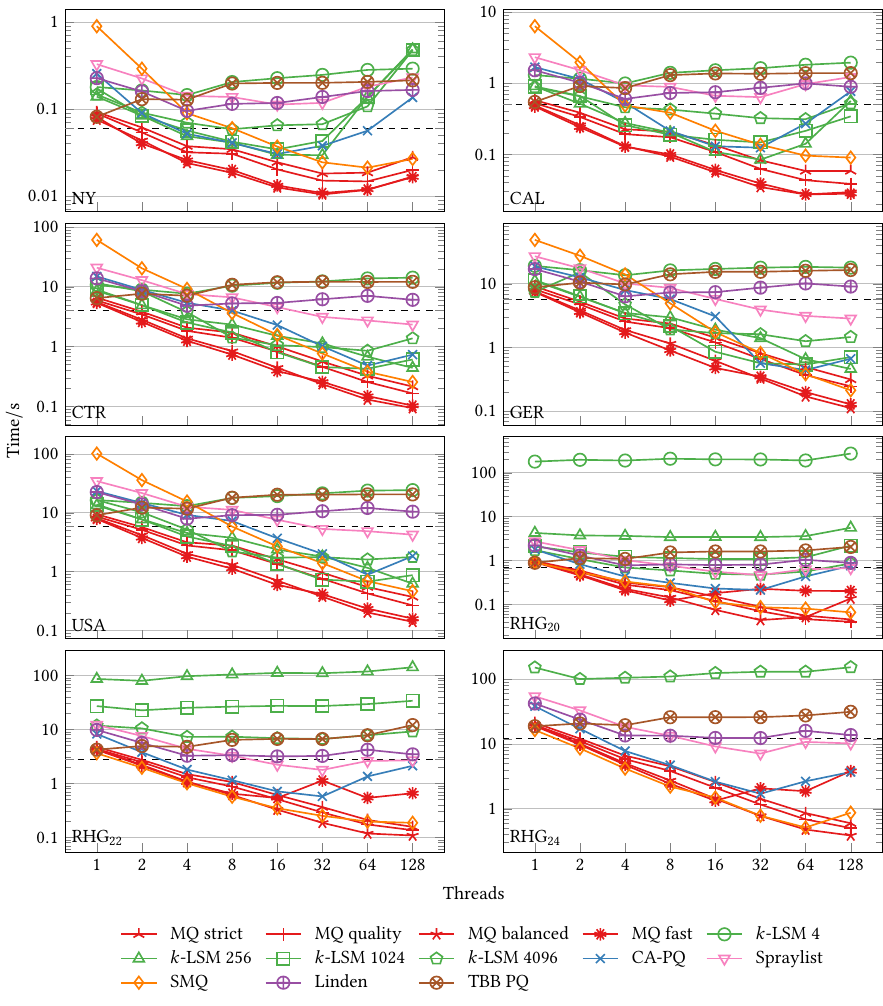}
	\caption{
		Time to solve the SSSP problem with different number of threads on machine \textsf{AMD}.
		The black dashed line is the solving time of the sequential algorithm.
        Missing data points are due to timeouts after \qty{300}{\second}.
	}
\end{figure}
\section{Knapsack}\label{appendix:knapsack}
\begin{figure}[ht]
	\includegraphics{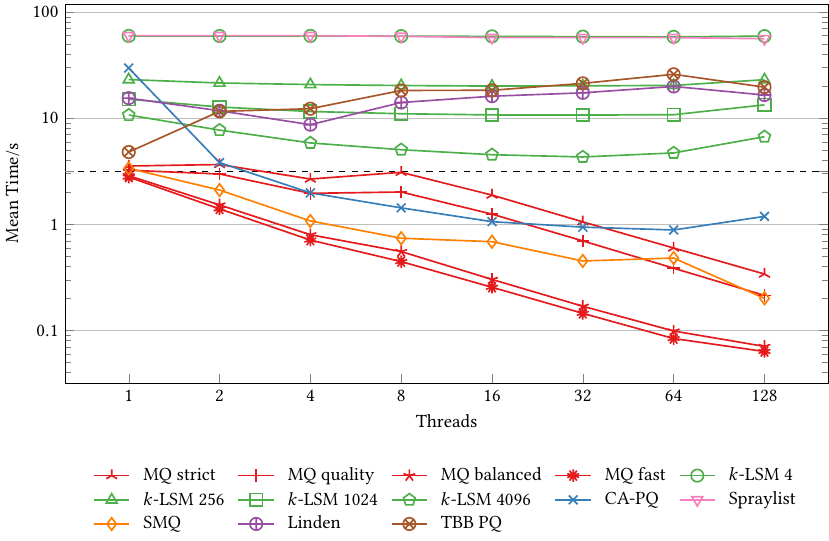}
	\caption{
		Arithmetic mean of times to solve each instance of the knapsack problem with different thread counts on machine \textsf{AMD}.
		If an instance is not solved within the time limit of \qty{60}{\second}, the time limit is assumed.
		The dashed black line indicates the sequential solving time.
	}
\end{figure}
\end{document}